\documentclass[11pt]{article}
\pdfoutput=1 
\usepackage[sc]{mathpazo}

\usepackage[margin=1in]{geometry}
\usepackage[english]{babel}
\usepackage[utf8x]{inputenc}
\usepackage[compact]{titlesec}

\usepackage{cmap}
\usepackage[T1]{fontenc}
\usepackage{bm}
\usepackage{amsfonts}
\usepackage{amsmath}
\usepackage{amssymb}
\usepackage{amsthm}
\usepackage{float}
\usepackage{graphics}

\usepackage{hyperref}
\usepackage[svgnames]{xcolor}
\hypersetup{colorlinks={true},urlcolor={blue},linkcolor={DarkBlue},citecolor=[named]{DarkGreen},linktoc=all}
\usepackage{natbib}

\usepackage{microtype}
\usepackage[capitalise,nameinlink,noabbrev]{cleveref}

\usepackage{doi}

\usepackage{booktabs} 
\usepackage{amsfonts,amsbsy,amsthm}
\usepackage{thmtools,thm-restate}
\usepackage{enumerate}
\usepackage{bbm}
\usepackage{nicefrac}
\usepackage{wrapfig}
\usepackage[justification = centering]{subcaption}
\usepackage{mathtools}
\usepackage{array,multirow}

\newcommand{\remove}[1]{}

  {\list{}{\leftmargin=#1\rightmargin=#1}\item[]}%
  {\endlist}

\newtheorem{theorem}{Theorem}[section]
\newtheorem{lemma}[theorem]{Lemma}
\newtheorem{definition}[theorem]{Definition}

\newtheorem{proposition}[theorem]{Proposition}
\newtheorem{corollary}[theorem]{Corollary}
\newtheorem{conjecture}[theorem]{Conjecture}

\newtheorem{observation}[theorem]{Observation}

\newtheorem{claim}[theorem]{Claim}

\newcommand{\envy}{\textit{envy}}

\newcommand{\wh}{\widehat}

\newcommand{\eps}{\epsilon}

\newcommand{\R}{\mathbb{R}}

\renewcommand{\i}{\mathbf{i}}

\renewcommand{\varepsilon}{\epsilon}

\renewcommand{\hat}{\wh}
\newcommand{\floor}[1]{\left\lfloor #1 \right\rfloor}

\definecolor{mygreen}{RGB}{80,180,0}
\definecolor{b2}{RGB}{51,153,255}
\definecolor{mycy2}{RGB}{255,51,255}

\newcommand{\Envy}{\mathsf{Envy}}

\colorlet{thechosenone}{red}
\colorlet{thechosentwo}{blue}
\makeatletter
\newcommand*{\RN}[1]{\expandafter\@slowromancap\romannumeral #1@}
\makeatother
\usepackage{lineno}

\usepackage{etoolbox}

\makeatletter

\newcommand{\define}[4][ignore]{%
  \ifstrequal{#1}{ignore}{}{
  \@namedef{thmtitle@#2}{#1}}%
  \@namedef{thm@#2}{#4}%
  \@namedef{thmtypen@#2}{lemma}%
  \newtheorem{thmtype@#2}[theorem]{#3}%
  \newtheorem*{thmtypealt@#2}{#3~\ref{#2}}%
}

\newcommand{\state}[1]{%
  \@namedef{curthm}{#1}
  \@ifundefined{thmtitle@#1}{
  \begin{thmtype@#1}
    }{
  \begin{thmtype@#1}[\@nameuse{thmtitle@#1}]
  }
    \label{#1}
    \@nameuse{thm@#1}
  \end{thmtype@#1}
  \@ifundefined{thmdone@#1}{
  \@namedef{thmdone@#1}{stated}%
  }{}
}

\newcommand{\restate}[1]{%
  \@namedef{curthm}{#1}
  \@ifundefined{thmtitle@#1}{
    \begin{thmtypealt@#1}
    }{
  \begin{thmtypealt@#1}[\@nameuse{thmtitle@#1}]
  }
    \@nameuse{thm@#1}
  \end{thmtypealt@#1}
  \@ifundefined{thmdone@#1}{
  \@namedef{thmdone@#1}{stated}%
  }{}
}

\newcommand{\thmlabel}[1]{
  \@ifundefined{thmdone@\@nameuse{curthm}}{\label{#1}
    }{\tag*{\eqref{#1}}}
}
\makeatother

\newif\ifcomments
\commentstrue
\ifcomments
\newcommand{\rik}[1]{{\textcolor{red}{Rik: { #1}}}}
\else
\newcommand{\rik}[1]{}
\fi
\newif\ifcomments
\commentstrue
\ifcomments
\newcommand{\justin}[1]{{\textcolor{blue}{Justin: { #1}}}}
\else
\newcommand{\justin}[1]{}
\fi
\newif\ifcomments
\commentstrue
\ifcomments
\newcommand{\vignesh}[1]{{\textcolor{orange}{Vignesh: { #1}}}}
\else
\newcommand{\vignesh}[1]{}
\fi
\newif\ifcomments
\commentstrue
\ifcomments
\newcommand{\rohit}[1]{{\textcolor{magenta}{Rohit: { #1}}}}
\else
\newcommand{\rohit}[1]{}
\fi
\newif\ifcomments
\commentstrue
\ifcomments
\newcommand{\hadi}[1]{{\textcolor{purple}{Hadi: { #1}}}}
\else
\newcommand{\hadi}[1]{}
\fi

\usepackage{tikz}
\usetikzlibrary{calc,arrows,matrix,decorations.pathreplacing}
\usepackage{paralist}
\usepackage{algorithm}
\usepackage[noend]{algpseudocode}
\usepackage{algorithmicx}

\usepackage[capitalise,noabbrev]{cleveref}
\Crefname{claim}{Claim}{Claims}
\Crefname{corollary}{Corollary}{Corollaries}
\Crefname{definition}{Definition}{Definitions}
\Crefname{example}{Example}{Examples}
\Crefname{lemma}{Lemma}{Lemmas}
\Crefname{property}{Property}{Properties}
\Crefname{proposition}{Proposition}{Propositions}
\Crefname{remark}{Remark}{Remarks}
\Crefname{theorem}{Theorem}{Theorems}

\title{Graphical House Allocation}
\author{
	\begin{tabular}{m{0.12\textwidth}m{0.12\textwidth}m{0.12\textwidth}m{0.12\textwidth}m{0.12\textwidth}m{0.12\textwidth}}
		\multicolumn{2}{c}{\textbf{Hadi Hosseini}} & \multicolumn{2}{c}{\textbf{Justin Payan}} & \multicolumn{2}{c}{\textbf{Rik Sengupta}}\\
		\multicolumn{2}{c}{\small{Penn State Univ.}} & \multicolumn{2}{c}{\small{Univ. of Massachusetts Amherst}} & \multicolumn{2}{c}{\small{Univ. of Massachusetts Amherst}}\\
		\multicolumn{2}{c}{\href{mailto:hadi@psu.edu}{\small{\texttt{hadi@psu.edu}}}} & \multicolumn{2}{c}{\href{mailto:jpayan@umass.edu}{\small{\texttt{jpayan@umass.edu}}}}
		& \multicolumn{2}{c}{\href{mailto:rsengupta@umass.edu}{\small{\texttt{rsengupta@umass.edu}}}}\\
		&&&&&\\
		\multicolumn{3}{c}{\textbf{Rohit Vaish}} & \multicolumn{3}{c}{\textbf{Vignesh Viswanathan}}\\
		\multicolumn{3}{c}{\small{IIT Delhi}} & \multicolumn{3}{c}{\small{Univ. of Massachusetts Amherst}}\\
		\multicolumn{3}{c}{\href{mailto:rvaish@iitd.ac.in}{\small{\texttt{rvaish@iitd.ac.in}}}} & \multicolumn{3}{c}{\href{mailto:vviswanathan@umass.edu}{\small{\texttt{vviswanathan@umass.edu}}}}\\
	\end{tabular}
}

\date{}

\begin{document}




\maketitle 


\begin{abstract}
The classical house allocation problem involves assigning $n$ houses (or items) to $n$ agents according to their preferences. A key criterion in such problems is satisfying some fairness constraints such as envy-freeness.
We consider a generalization of this problem wherein the agents are placed along the vertices of a graph (corresponding to a social network), and each agent can only experience envy towards its neighbors. Our goal is to minimize the \emph{aggregate} envy among the agents as a natural fairness objective, i.e., the sum of all pairwise envy values over all edges in a social graph.

When agents have identical and \emph{evenly-spaced} valuations, our problem reduces to the well-studied problem of \emph{linear arrangements}. For identical valuations with possibly uneven spacing, we show a number of deep and surprising ways in which our setting is a departure from this classical problem. More broadly, we contribute several structural and computational results for various classes of graphs, including NP-hardness results for disjoint unions of paths, cycles, stars, or cliques, and fixed-parameter tractable (and, in some cases, polynomial-time) algorithms for paths, cycles, stars, cliques, and their disjoint unions.
Additionally, a conceptual contribution of our work is the formulation of a structural property for disconnected graphs that we call \emph{separability} which results in efficient parameterized algorithms for finding optimal allocations.
\end{abstract}

\section{Introduction}\label{sec:intro}
The house allocation problem has attracted interest from the computer science and multiagent systems communities for a long time. The classical problem deals with assigning $n$ houses to a set of $n$ agents with (possibly) different valuations over the houses. It is often desirable to find assignments that satisfy some economic property of interest. In this work, we focus on the well-motivated economic notion of \emph{fairness}, and in particular, study the objective of minimizing the \emph{aggregate envy} among the agents.

Despite the historical interest in this problem, to the best of our knowledge, the house allocation problem has not been studied thoroughly over \emph{graphs}, a setting in which the agents are placed on the vertices of an undirected graph $G$ and each agent's potential envy is only towards its neighbors in $G$ \citep{graphicalonesided, beynier2018localenvy}.



Incorporating the social structure over a graph enables us to capture the underlying restrictions of dealing with partial information, which is representative of constraints in many real-world applications.   
Thus, the classical house allocation problem is the special case of our problem when the underlying graph is complete.


Our work is in line with recent literature on examining various problems in computational social choice on social networks, including voting~\citep{doucette2019inferring,tsang2015voting,grandi2017social}, fair division~\citep{abebe2017fair, bredereck2022envy}, and hedonic games~\citep{peters2016graphical,igarashi2016hedonic}. By focusing on graphs, we aim to gain insights into how the structure of the social network impacts fairness in house allocation. We focus primarily on \emph{identical} valuation functions and show that even under this seemingly strong restriction, the problem is computationally hard yet structurally rich. 
We provide a series of observations and insights about graph structures that help identify, and in some cases overcome, these computational bottlenecks.

\subsection{Overview and Our Contributions}

We assume that the agents are placed at the vertices of a graph representing a social network, and that they have identical valuation functions over the houses. Our objective is to find an allocation of the houses among the agents to minimize the total envy in the graph. We call this the \emph{graphical house allocation problem}.

This is a beautiful combinatorial problem in its own right, as it can be restated as the problem where, given an undirected graph and a multiset of nonnegative numbers, the numbers need to be placed on the vertices in a way that minimizes the sum of the edgewise absolute differences.



In Section \ref{sec:prelims}, we present the formal model and set up some preliminaries, including the connection between the graphical house allocation problem and the \emph{minimum linear arrangement} problem, which has several notable similarities and differences.

In Section \ref{sec:lowerbounds}, we present computational lower bounds and inapproximability results for the problem, even for very simple graphs. In particular, we show NP-hardness even when the graph is a disjoint union of paths, cycles, stars, or cliques, which all have known polynomial-time algorithms for linear arrangements.

In Section \ref{sec:connected}, we focus on connected graphs and completely characterize optimal allocations when the graph is a path, cycle, star, or a complete bipartite graph. We also prove a technically involved structural result for \textit{rooted binary trees}, and discuss general trees.

In Section \ref{sec:disconnected}, we focus on disconnected graphs, starting with a fundamental difference between graphical house allocation and linear arrangements, motivating our definitions of \emph{separable}, \emph{strongly separable}, and \emph{inseparable} disconnected graphs. 
%
%
We employ these characterizations to prove algorithmic results for a variety of graphs. 
In particular, we show that disjoint unions of paths, cycles, stars, and equal-sized cliques are strongly separable and develop natural fixed parameter tractable algorithms for these graphs. 
Moreover, we show that disjoint unions of arbitrary cliques satisfy separability~(but not strong separability) and admit XP algorithms.

Finally, in Section \ref{sec:conclusion}, we wrap up with a concluding discussion.

In the interest of space, we mainly provide proof sketches and defer the details to the appendices.

\subsection{Related Work}
House allocation has been traditionally studied in the economics literature under the \emph{housing market} model, where agents enter the market with a house (or an endowment) each and are allowed to engage in cyclic exchanges~\citep{shapley1974cores}. This model has found important practical applications, most notably in kidney exchange~\citep{AS99house,roth2004kidney}.

While the initial work on house allocation focused on the economic notions of core and strategyproofness~\citep{svensson1999strategy}, subsequent work has explored \emph{fairness} issues. 
\citet{gsvfairhouse} study the house allocation problem under ordinal preferences (specifically, weak rankings) and provide a polynomial-time algorithm for determining the existence of an envy-free allocation. By contrast, the problem becomes NP-hard when agents' preferences are specified as a set of pairwise comparisons~\citep{kamiyama2021envy}. 
%
\citet{kmsfairhouse} study house allocation under cardinal preferences (similar to our work) and examine the complexity of finding a ``fair'' assignment for various notions of fairness such as proportionality, equitability, and minimizing the number of envy-free agents (they do not consider aggregate envy). They show that the latter problem is hard to approximate under general valuations, and remains NP-hard even for the restricted case of binary valuations. 
For binary valuations, the problem of finding the largest envy-free partial matching has also been studied~\citep{aigner2022envy}. 
%
%

Recent studies have considered \emph{graphical} aspects of house allocation (similar to us), though with different objectives.  \citet{graphicalonesided} consider house allocation under externalities and study various kinds of stability-based objectives. \citet{beynier2018localenvy}, whose work is perhaps closest to ours, study \emph{local envy-freeness} in house allocation, which entails checking the existence of an allocation with no envy along any edge of the graph. They delineate the computational complexity of this problem with respect to various graph parameters such as maximum and minimum degree, number of disjoint cliques, and size of minimum vertex cover. Notably, their model involves agents with \emph{possibly distinct} ordinal preferences and only the zero-envy condition, which makes their results---hardness as well as algorithmic---not directly comparable with ours.


There is also a growing literature on fair allocation of indivisible objects among agents who are part of a social network. \citet{bredereck2022envy} present fixed-parameter tractability results, mainly parametrized by the number of agents, though they leave results using graph structure to future work.
\citet{eiben2020parameterized} extend these results, showing a number of parametrized complexity results relating the treewidth, cliquewidth, number of agent types, and number of item types to the complexity of determining if an envy-free allocation exists on a graph.
This line of work again focuses on deciding if envy-free allocations exist, not minimizing envy. 
Other works seek to obtain envy-free allocations, maximum welfare allocations, or other objectives by swapping objects along a graphical structure \citep{beynier2018fairness, lange2019optimizing, gourves2017object, ito2022reforming}. Their objectives differ from ours, though their work is similar in spirit.







\section{Preliminaries}\label{sec:prelims}

We use $[t]$ to denote the set $\{1, 2, \dots, t\}$.
There is a set of $n$ {\em agents} $N = \{1, 2, \ldots, n\}$ and $n$ {\em houses} $H = \{h_1, h_2, \dots, h_n\}$ (often called {\em items}). Each agent $i$ has a {\em valuation function} $v_i: H \rightarrow \R_{\geq 0}$; $v_i(h)$ indicates agent $i$'s value for house $h \in H$.

An {\em allocation} $\pi$ is a bijective mapping from agents to houses. For each $i \in N$, $\pi(i)$ is the house allocated to agent $i$ under the allocation $\pi$, and $v_i(\pi(i))$ is its utility. 

Given a problem instance consisting of agents and houses, our goal is to generate an allocation $\pi$ that is ``fair'' to all the agents, for some reasonable definition of fairness. A natural way to define fairness is using {\em envy}: an agent $i$ is said to envy agent $j$ under allocation $\pi$ if $v_i(\pi(i)) < v_i(\pi(j))$. While we would ideally like to find \emph{envy-free} allocations, this may not always be possible --- consider a simple example with two agents and two houses but one of the houses is valued at $0$ by both agents. Therefore, we instead focus on the magnitude of envy that agent $i$ has towards agent $j$, for a fixed allocation $\pi$. This is defined as $\envy_{i}(\pi, j) := \max\{v_i(\pi(j)) - v_i(\pi(i)), 0\}$. 

We define an undirected graph $G = (N, E)$ over the set of agents, which represents the underlying social network. Our goal is to compute an allocation that minimizes the {\em total envy} along the edges of the graph, defined as $\Envy(\pi, G) := \sum_{(i, j) \in E}(\envy_i(\pi, j) + \envy_j(\pi, i))$; note that edges are unordered. 
An allocation $\pi^{\ast}$ that minimizes the total envy is referred to as a {\em minimum envy allocation}. 
When there are multiple minimum envy allocations, we abuse notation slightly and use $\pi^{\ast}$ to denote any of them.

When the graph $G$ is a complete graph $K_n$, a minimum envy allocation can be computed in polynomial time by means of a reduction to a bipartite minimum-weight matching problem (for a proof of this, see Proposition \ref{prop:knarbitrary} in Appendix \ref{apdx:prelims}). 
However, it is known that for several other simple graphs like paths and matchings, computing a minimum envy allocation is NP-complete \citep{beynier2018localenvy}. Given this computational intractability, we therefore explore a natural restriction of the problem, when all agents have identical valuations, to gain insights into the computational and structural aspects of fairness in social networks. More formally, we assume that there is a fixed valuation function $v$ such that $v_i = v$ for all $i \in N$. Identical valuations capture a natural aspect of real-world housing markets wherein the actual values or prices of the houses are fixed, independent of whom they are being assigned to.

When all agents have the same valuation function $v$, the total envy of an allocation $\pi$ along the edges of a graph $G = (N, E)$ can be written as $\Envy(\pi, G) = \sum_{(i, j) \in E} |v(\pi(i)) - v(\pi(j))|$. This formulation also gives a new definition for envy along an edge $e = (i, j) \in E$ as $\envy_e(\pi) = |v(\pi(i)) - v(\pi(j))|$. Note that this value is still equal to $\envy_i(\pi, j) + \envy_j(\pi, i)$, as one of those terms is zero under the assumption of identical valuations.

We note that when $G$ is $K_n$, under identical valuations, an optimal allocation is trivially computable, as all allocations are equivalent.




For the rest of this paper, we will assume without loss of generality that the house values are all distinct (see Lemma \ref{lem:valspositive} in Appendix \ref{apdx:prelims}). In particular, every agent's valuation function (denoted by $v$) gives each house a unique nonnegative value, with $v(h_1) < v(h_2) < \dots < v(h_n)$. We will say $h_1 \prec h_2$ to mean $v(h_1) < v(h_2)$.

For an allocation $\pi$ and a subset $N' \subseteq N$, we will refer to the set of houses received by $N'$ as $\pi(N')$. If $G'$ is a subgraph of $G$, we will use $\pi(G')$ in the same way.

We will use $G_1 + G_2$ to mean the disjoint union of $G_1$ and $G_2$.

\begin{definition}\label{valint}
For an instance of the graphical house allocation problem, the \emph{valuation interval} is defined as the closed interval $[v(h_1), v(h_n)] \subset \mathbb{R}^+$ with each $v(h_k)$ marked.
\end{definition}
The motivation for Definition \ref{valint} is as follows. For an arbitrary allocation $\pi$, for each edge $e = (i, j) \in E$, we can draw a line segment from $v(\pi(i))$ to $v(\pi(j))$. This line segment has length $|v(\pi(i)) - v(\pi(j))| = \envy_e(\pi)$. Therefore, $\Envy(\pi, G)$ is simply the sum of the lengths of all the line segments we draw in this way. An optimal allocation $\pi^\ast$ is any allocation that attains this minimum sum. See Figure \ref{fig:valuationlineexample} for an example of a valuation interval, together with a graph $G$, and a particular allocation on $G$ depicted under the valuation interval.

\begin{figure}[ht]
    \centering
    \small
    \begin{tikzpicture}
                \tikzset{mynode/.style = {shape=circle,draw,inner sep=1.5pt}}
                \tikzset{edge/.style = {solid}}
                \node[mynode] (1) at (0,-1.1) {};
                \node (1a) at (0,-1.4) {\color{black}{$h_1$}};
                \node[mynode] (2) at (1,-1.1) {};
                \node (2a) at (1,-1.4) {\color{black}{$h_4$}};
                \node[mynode] (3) at (2,-1.1) {};
                \node (3a) at (2,-1.4) {\color{black}{$h_2$}};
                \node[mynode] (4) at (2,-0.3) {};
                \node (4a) at (2,0) {\color{black}{$h_5$}};
                \node[mynode] (5) at (0.5,-0.3) {};
                \node (5a) at (0.5,0) {\color{black}{$h_3$}};
                \draw (1) -> (2);
                \draw (2) -> (3);
                \draw (3) -> (4);
                \draw (1) -> (5);
                \draw (2) -> (5);
                \edef\myvar{1}
                \def\valuations{{1,2,4,5,6}}
                \def\locations{3.5,4.5,6,7,8}
                \foreach \x in {7,9,12,14,16}
                    {
                        \coordinate (A\x) at ($(\x/2,0)$) {};
                        \draw ($(A\x)+(0,3pt)$) -- ($(A\x)-(0,3pt)$);
                        \node (\myvar) at ($(A\x)+(0,3ex)$) {{$h_{\pgfmathprintnumber{\myvar}}$}};
                        \pgfmathparse{\myvar+1}
                        \xdef\myvar{\pgfmathresult}
                    }
                \draw[line width=0.5 mm] (3.5,0) -- (8,0);
                \node at (3.5,-2ex) {\color{red}{1}};
                \node at (4.5,-2ex) {\color{red}{2}};
                \node at (6,-2ex) {\color{red}{4}};
                \node at (7,-2ex) {\color{red}{5}};
                \node at (8,-2ex) {\color{red}{6}};
                \node at (5.7,-1.4) {\small{\color{thechosenone}{Envy=15}}};
                \draw[color=thechosenone,line width=0.5 mm] (3.5,-0.5) -- (6,-0.5);
                \draw[color=thechosenone,line width=0.5 mm] (6,-0.65) -- (7,-0.65);
                \draw[color=thechosenone,line width=0.5 mm] (3.5,-0.8) -- (7,-0.8);
                \draw[color=thechosenone,line width=0.5 mm] (4.5,-0.95) -- (7,-0.95);
                \draw[color=thechosenone,line width=0.5 mm] (4.5,-1.1) -- (8,-1.1);
    \end{tikzpicture}
    \caption{(Left) A graph $G$ on five agents along with a particular allocation $\pi$. The valuations are identical and are given by $\vec{v} = (1, 2, 4, 5, 6)$. (Right) The valuation interval is shown via the thick horizontal line in black. The five line segments in {\color{thechosenone}red}  denote the envy along the five edges of the graph $G$. The total length of these line segments is $\Envy(\pi, G) = 15$.}
    \label{fig:valuationlineexample}
\end{figure}
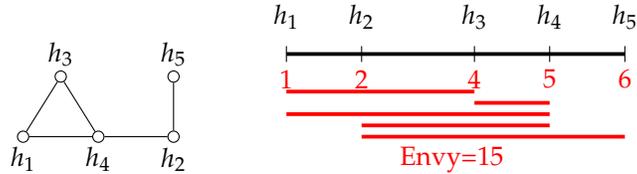

A subset of houses $H' = \{h_{i_1}, \ldots, h_{i_k}\} \subseteq H$ with $h_{i_1} \prec \ldots \prec h_{i_k}$ is called \emph{contiguous} if there is no house $h' \in H\setminus H'$ with $h_{i_1} \prec h' \prec h_{i_k}$. Geometrically, the values in $H'$ form an uninterrupted sub-interval of the valuation interval, with no value outside of $H'$ appearing inside that sub-interval. In Figure \ref{fig:valuationlineexample} above, the subsets $\{h_1, h_2, h_3\}$ and $\{h_5\}$ are contiguous, whereas the subsets $\{h_1, h_2, h_5\}$ and $\{h_3, h_5\}$ are not.

We will sometimes interchangeably talk about allocating $h_i$ and allocating $v(h_i)$ to an agent. For readability, we will also sometimes refer to houses as being on the valuation line, rather than their values. The meaning should be unambiguous from the context.

\subsection{Connection to the Linear Arrangement Problem}\label{sec:mla}


The \emph{minimum linear arrangement problem} is the problem where, given an undirected $n$-vertex graph $G = (V, E)$, we want to find a bijective function $\pi : V \to [n]$ that minimizes $\sum_{(i, j) \in E}|\pi(i) - \pi(j)|$. Note that this is a special case of our problem, when $H = [n]$; in other words, the minimum linear arrangement problem is a special case of the graphical house allocation problem where the valuation interval has evenly spaced values.


Note that we may assume without loss of generality that the underlying graph $G$ in any instance of the minimum linear arrangement problem is connected. This is a consequence of the following folklore observation, whose proof we omit (see, for instance,~\citep{seidvasser}).

\begin{restatable}[\citet{seidvasser}]{proposition}{propmlawlog}\label{prop:mlawlog}
If $G$ is any instance of the linear arrangement problem, then at least one optimal solution assigns contiguous subsets of $[n]$ to the connected components of $G$.
\end{restatable}

We will see in Section \ref{sec:disconnected} that in the graphical house allocation problem, it no longer suffices to only consider connected graphs.

It is known that finding a minimum linear arrangement is NP-hard for general graphs~\citep{mlahard}, with a best known running time of $O(2^nm)$, where $|V| = n, |E| = m$. In fact, the problem remains NP-hard even for bipartite graphs~\citep{mlabinaryhard}. However, the minimum linear arrangement problem can be solved in polynomial time on paths, cycles, stars, wheels, and trees. In fact, an optimal solution for any tree can be computed in time $O(n^{\log_23})$~\citep{mlatrees}.


One of the fundamental difficulties with the minimum linear arrangement problem is that there are no natural approaches towards proving lower bounds on the objective: the best known general lower bounds are trivial ones based on degrees or edge counts. All other known lower bounds are graph-specific and not easy to generalize.

\section{Hardness and Lower Bounds}\label{sec:lowerbounds}

In this section, we prove lower bounds on the graphical house allocation problem, showing that the problem is hard to solve in general.

First of all, note that our problem is NP-complete~(for arbitrary graphs and valuations), because even the linear arrangement problem is NP-complete, as stated in Section \ref{sec:prelims}, and the graphical house allocation problem is at least as hard as that problem.

We now show a different proof that the problem is NP-complete, which will result in inapproximability results as well. First we need a definition.

\begin{definition}
The \emph{\textsc{Minimum Bisection Problem}} asks, for an $n$-vertex graph $G$ and a natural number $k$, if there is a partition of $V(G)$ into two parts 
of size $n/2$, with at most $k$ edges crossing the cut.
\end{definition}

The \textsc{Minimum Bisection Problem} is a known NP-complete problem~
\citep{mlahard}. Furthermore, it is also known to be hard to approximate efficiently, a fact that is useful in light of the following observation.

\begin{restatable}{theorem}{npcompletebipartition}\label{thm:npcompletebipartition}
There is a linear time reduction from the \textsc{Minimum Bisection Problem} to the graphical house allocation problem with identical valuations.
\end{restatable}
\begin{proof}[Proof Sketch]
Given an instance $\langle G, k\rangle$ of the \textsc{Minimum Bisection Problem}, we construct an instance of the house allocation problem by creating a cluster of $n/2$ values concentrated around $0$, and a cluster of $n/2$ values concentrated around $1$ on the valuation interval. Then, $G$ has an allocation with total envy at most $k + \varepsilon$ (for sufficiently small $\varepsilon$) if and only if $\langle G, k\rangle \in$ \textsc{Minimum Bisection Problem}. This reduction is clearly linear time.
\end{proof}

It follows immediately that the inapproximability results for minimum bisection carry over to the graphical house allocation problem. In particular, for any fixed constant $\varepsilon > 0$, unless P = NP there is no polynomial-time algorithm that can approximate the optimal total envy under the house allocation problem within an additive term of $n^{2 - \varepsilon}(v(h_n) - v(h_1))$~\citep{bj92}. Additionally, the house allocation problem has no PTAS unless NP has randomized algorithms in subexponential time~\citep{khot04}. These hardness results suggest that our problem is hard to approximate \emph{even} with identical valuations.

Finally, we show that graphical house allocation is NP-complete even on simple instances of graphs which are solvable in linear time in the case of linear arrangements, such as disjoint unions of paths, cycles, cliques, or stars (and any combinations of them).

\begin{restatable}[Hardness of Disjoint Unions]{theorem}{thmdisjointnpcomplete}\label{thm:disjointunions}
Let $\mathcal{A}$ be any collection of connected graphs, such that there is a polynomial time one-to-one mapping from each nonnegative integer $t$ (given in unary) to a graph in $\mathcal{A}$ of size $t$. Let $\mathcal{G}$ be the class of graphs whose members are the finite sub-multisets of $\mathcal{A}$ (as connected components). Then, finding a minimum envy house allocation is NP-hard on the class $\mathcal{G}$.
\end{restatable}

\begin{proof}[Proof Sketch]
We reduce from \textsc{Unary Bin Packing}. Recall that this problem asks if, given a set $I$ of items with sizes $s(i)$ for $i \in I$, a bin capacity $B$, and an integer $k$, all given \emph{in unary}, whether there exists a \emph{packing} of all the items into at most $k$ bins. This problem is known to be NP-complete (see, for instance,~\cite{binpackinghardness}).

Given an arbitrary instance of \textsc{Unary Bin Packing}, we create an instance of graphical house allocation by having $k$ equispaced clusters of width $\varepsilon$ separated by intervals of size $C$ (for sufficiently large $C$), with each cluster containing $B$ values. Our graph $G$ is a disjoint union of the graphs in $\mathcal{A}$ that form the image of the sizes $s(i)$ over all $i \in I$. Note that $G \in \mathcal{G}$. Then, the given instance is in \textsc{Unary Bin Packing} if and only if the graphical house allocation instance has an allocation with envy less than $C$, provided $\varepsilon$ is small enough.
\end{proof}

\begin{corollary}\label{cor:disjointnp}
The house allocation problem under identical valuations is NP-complete on:
\begin{inparaenum}[(a)]
    \item disjoint unions of arbitrary paths,
    \item disjoint unions of arbitrary cycles,
    \item disjoint unions of arbitrary stars, and
    \item disjoint unions of arbitrary cliques.
\end{inparaenum}
\end{corollary}
In Section \ref{sec:disconnected} we show that despite the hardness suggested by Corollary \ref{cor:disjointnp}, it is possible to exploit a  structural property to develop FPT algorithms for the first three problems.

\section{Connected Graphs}\label{sec:connected}

In this section, we characterize optimal house allocations when the underlying graph $G$ is a star, path, cycle, complete bipartite graph, or rooted binary tree. We also provide some observations when $G$ is any (arbitrary) tree.

\begin{figure*}
    \begin{subfigure}[b]{0.23\textwidth}
    \centering
    \small
    \begin{tikzpicture}[,mycirc/.style={circle,fill=white, draw = black,minimum size=0.25cm,inner sep = 3pt}]
        \tikzset{mynode/.style = {shape=circle,draw,inner sep=1.5pt}}
        \tikzset{edge/.style = {solid}}
        
        \node[draw, mynode] at (360:0mm) (center) {};
        \foreach \i [count=\ni from 0] in {1, 2, 3, 4, 5}{
          \node[draw, mynode] at (\ni*72 + 90:1cm) (\ni) {};
          \draw (center)--(\ni);
        }
    \end{tikzpicture}
    \caption{The star $K_{1, 5}$}
    \label{subfig:star}
    \end{subfigure}
    \begin{subfigure}[b]{0.23\textwidth}
    \centering
    \small
    \begin{tikzpicture}[,mycirc/.style={circle,fill=white, draw = black,minimum size=0.75cm,inner sep = 3pt}]
        \tikzset{mynode/.style = {shape=circle,draw,inner sep=1.5pt}}
        \tikzset{edge/.style = {solid}}
        \node[mynode] (1) at (0,0) {};
        \node[mynode] (2) at (0.5,1.5) {};
        \node[mynode] (3) at (1,0) {};
        \node[mynode] (4) at (1.5,1.5) {};
        \node[mynode] (5) at (2, 0) {};
        \node[mynode] (6) at (2.5, 1.5) {};
        \node[mynode] (7) at (3, 0) {};
        \draw (1) -> (2);
        \draw (2) -> (3);
        \draw (3) -> (4);
        \draw (4) -> (5);
        \draw (5) -> (6);
        \draw (6) -> (7);
    \end{tikzpicture}
    \caption{The path $P_7$}
    \label{subfig:path}
    \end{subfigure}
    \begin{subfigure}[b]{0.23\textwidth}
    \centering
    \small
    \begin{tikzpicture}[,mycirc/.style={circle,fill=white, draw = black,minimum size=0.75cm,inner sep = 3pt}]
        \tikzset{mynode/.style = {shape=circle,draw,inner sep=1.5pt}}
        \tikzset{edge/.style = {solid}}
        
        \foreach \i [count=\ni from 1] in {1, 2, 3, 4, 5}{
          \node[draw, mynode] at (\ni*72 +18:1cm) (\ni) {};
        }
        
        \draw (1) -> (2);
        \draw (2) -> (3);
        \draw (3) -> (4);
        \draw (4) -> (5);
        \draw (5) -> (1);
    \end{tikzpicture}
    \caption{The cycle $C_5$}
    \label{subfig:cycle}
    \end{subfigure}
    \begin{subfigure}[b]{0.23\textwidth}
    \centering
    \small
    \begin{tikzpicture}[,mycirc/.style={circle,fill=white, draw = black,minimum size=0.75cm,inner sep = 3pt}]
        \tikzset{mynode/.style = {shape=circle,draw,inner sep=1.5pt}}
        \tikzset{edge/.style = {solid}}
        \node[mynode] (1) at (0,0) {};
        \node[mynode] (2) at (1,0) {};
        \node[mynode] (3) at (2,0) {};
        \node[mynode] (4) at (0,1.5) {};
        \node[mynode] (5) at (1,1.5) {};
        \node[mynode] (6) at (2, 1.5) {};
        \draw (1) -> (4);
        \draw (1) -> (5);
        \draw (1) -> (6);
        \draw (2) -> (4);
        \draw (2) -> (5);
        \draw (2) -> (6);
        \draw (3) -> (4);
        \draw (3) -> (5);
        \draw (3) -> (6);
    \end{tikzpicture}
    \caption{The graph $K_{3, 3}$}
    \label{subfig:bipartite}
    \end{subfigure}
    \caption{Examples of characterized connected graphs}
    \label{fig:examples}
\end{figure*}
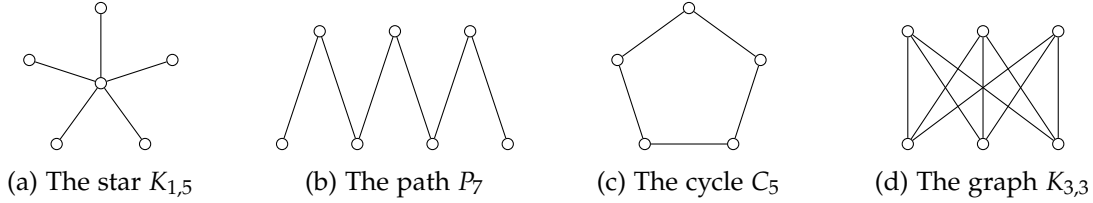

\subsection{Stars}

Consider the \emph{star graph} $K_{1, n}$.
        


\begin{restatable}{theorem}{thmstar}\label{thm:star}
If $G$ is the star $K_{1,n}$, then the minimum envy allocation $\pi^\ast$ under identical valuations corresponds to:
\begin{itemize}
    \item for even $n$, putting the unique median in the center of the star, and all the houses on the spokes in any order; the value of the envy is $\sum_{i > n/2 + 1}v(h_i) - \sum_{i \leq n/2}v(h_i)$.
    \item for odd $n$, putting either of the medians in the center of the star, and all other houses on the spokes in any order; the value of the envy for either median is $\sum_{i > (n+2)/2}v(h_i) - \sum_{i < (n+2)/2}v(h_i)$.
\end{itemize}
\end{restatable}

\begin{proof}
The proof is a restatement of the well-known fact that in any multiset of real numbers, the sum of the $L_1$-distances is minimized by the median of the multiset. It is easy to verify that for even $n$, both medians yield the same value.
\end{proof}

\subsection{Paths}

Consider the \emph{path graph} $P_n$.



\begin{restatable}{theorem}{thmpath}\label{thm:path}
If $G$ is the path graph $P_n$, then the minimum envy allocation $\pi^\ast$ under identical valuations attains a total envy of $v(h_n) - v(h_1)$, is \emph{unique} (up to reversing the values along the path), and corresponds to placing the houses in sorted order along $P_n$.
\end{restatable}

\begin{proof}[Proof Sketch]
The minimum and the maximum have to go on some two vertices. The subpath from the minimum to the maximum has to have envy at least the difference between the two.
\end{proof}

\subsection{Cycles}

Now, consider the \emph{cycle graph} $C_n$.

        
        


\begin{restatable}{theorem}{thmcycle}\label{thm:cycle}
If $G$ is the cycle graph $C_n$, then any minimum envy allocation $\pi^\ast$ under identical valuations attains a total envy of $2(v(h_n) - v(h_1))$, and corresponds to the following: place $h_1$ and $h_n$ arbitrarily on any two vertices of the cycle, and then place the remaining houses so that each of the two paths from $h_1$ to $h_n$ along the cycle consists of houses in sorted order.
\end{restatable}

\begin{proof}[Proof Sketch]
In any allocation, there are two internally vertex-disjoint subpaths from $h_1$ to $h_n$ on a cycle. Each of those subpaths has an envy of at least the difference between the two values.
\end{proof}

\begin{restatable}{corollary}{corcycle}
For $n \geq 3$, the number of optimal allocations along the cycle $C_n$ is $2^{n-3}$, up to rotations and reversals.
\end{restatable}

Perhaps slightly non-obviously
, the proofs of Theorems \ref{thm:path} and \ref{thm:cycle} can be seen as purely geometric arguments using the valuation interval. To see this, consider the path $P_n$, and take any allocation $\pi$ that does not satisfy the form stated in Theorem \ref{thm:path}, and consider how the allocation looks on the valuation interval. First, observe that every sub-interval of the valuation interval between consecutive houses needs to be covered by some line segment from the allocation. Otherwise, there would be no edge with a house from the left to a house from the right of the sub-interval, which is impossible as $P_n$ is connected. But the only way to meet this lower bound of one line segment for each sub-interval of the valuation interval is to sort the houses along the path. 
The allocation looks as follows on the valuation interval.
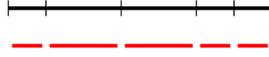
\begin{figure}[H]
    \centering
    \small
    \begin{tikzpicture}
    \xdef\eps{0.05}
        \foreach \x in {1,2,4,6,7,8}
            {
                \coordinate (A\x) at ($(\x/2,0)$) {};
                \draw ($(A\x)+(0,3pt)$) -- ($(A\x)-(0,3pt)$);
            }
            \draw[line width=0.5 mm] (0.5,0) -- (4,0);
            \draw[color=thechosenone,line width=0.5 mm] (0.5+\eps,-0.5) -- (1-\eps,-0.5);
            \draw[color=thechosenone,line width=0.5 mm] (1+\eps,-0.5) -- (2-\eps,-0.5);
            \draw[color=thechosenone,line width=0.5 mm] (2+\eps,-0.5) -- (3-\eps,-0.5);
            \draw[color=thechosenone,line width=0.5 mm] (3+\eps,-0.5) -- (3.5-\eps,-0.5);
            \draw[color=thechosenone,line width=0.5 mm] (3.5+\eps,-0.5) -- (4-\eps,-0.5);
    \end{tikzpicture}
\end{figure}

The geometric argument for cycles is similar (with an allocation illustrated below).

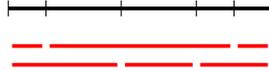
\begin{figure}[H]
    \centering
    \small
    \begin{tikzpicture}
    \xdef\eps{0.05}
        \foreach \x in {1,2,4,6,7,8}
            {
                \coordinate (A\x) at ($(\x/2,0)$) {};
                \draw ($(A\x)+(0,3pt)$) -- ($(A\x)-(0,3pt)$);
            }
            \draw[line width=0.5 mm] (0.5,0) -- (4,0);
            \draw[color=thechosenone,line width=0.5 mm] (0.5+\eps,-0.5) -- (1-\eps,-0.5);
            \draw[color=thechosenone,line width=0.5 mm] (1+\eps,-0.5) -- (3.5-\eps,-0.5);
            \draw[color=thechosenone,line width=0.5 mm] (3.5+\eps,-0.5) -- (4-\eps,-0.5);
            \draw[color=thechosenone,line width=0.5 mm] (0.5+\eps,-0.75) -- (2-\eps,-0.75);
            \draw[color=thechosenone,line width=0.5 mm] (2+\eps,-0.75) -- (3-\eps,-0.75);
            \draw[color=thechosenone,line width=0.5 mm] (3+\eps,-0.75) -- (4-\eps,-0.75);
    \end{tikzpicture}
\end{figure}

\subsection{Complete Bipartite graphs}

Let us start with the complete bipartite graph $K_{r, r}$ $(r \ge 1)$ where both parts have equal size. 

\begin{restatable}{theorem}{thmbipartiteequal}\label{thm:bipartite-equal}
When $G$ is the graph $K_{r, r}$, the minimum envy allocation $\pi^{\ast}$ has the following property: for every $i \in [r]$ the houses $\{h_{2i-1}, h_{2i}\}$ cannot be allocated to agents in the same side of the bipartite graph. Moreover, all allocations which satisfy this property have the same (optimal) envy.
\end{restatable}

\begin{proof}[Proof Sketch]
Consider an allocation $\pi$ where the stated property does not hold and consider the smallest $i$ where the property is violated. In this case we can find a simple swap that can improve envy. Assume $h_{2i-1}$ and $h_{2i}$ are allocated to one part. If $k$ is the least index greater than $2i$ such that $h_k$ is not allocated to an agent in the same part, then swapping $h_k$ and $h_{k-1}$ leads to a reduction in envy. The exact calculation is quite technical; see Appendix \ref{apdx:connected}.
\end{proof}

Note that this implies a straightforward linear-time algorithm to compute a minimum envy allocation for the graph $K_{r, r}$.

We can now generalize this result to complete bipartite graphs where the two parts have unequal size. Due to the similarity of the two proofs, we omit the proof sketch.
\begin{restatable}{theorem}{thmbipartiteunequal}\label{thm:bipartite-unequal}
When $G$ is the graph $K_{r, s}$ ($r > s$), the minimum envy allocation $\pi^\ast$ has the following property: 
\begin{itemize}
    \item If $ r - s =: 2m$ is even, then the first and last $m$ houses are allocated to the larger part, and for all $i \in [s]$, the houses $h_{m + 2i - 1}$ and $h_{m +2i}$ are allocated to different parts.
    \item If $r - s =: 2m + 1$ is odd, then the first $m$ and last $m + 1$ houses are allocated to the larger part. For all $i \in [s]$, the houses $h_{m + 2i - 1}$ and  $h_{m
    + 2i}$ are allocated to the larger and smaller parts respectively.
\end{itemize}
Moreover, all allocations which satisfy this property have the same (optimal) envy.
\end{restatable}

\begin{restatable}{corollary}{corbipartite}\label{cor:bipartite-solutions}
For any complete bipartite graph $K_{r, s}$ ($r \ge s$),
\begin{itemize}
    \item If $r - s$ is even, there are $2^s$ optimal allocations;
    \item If $r - s$ is odd, there is exactly one optimal allocation,
\end{itemize}
up to permutations over allocations to the same side of the graph.
\end{restatable}

It is easy to see that our linear time algorithm generalizes to general complete bipartite graphs as well. 
We note here that Theorem \ref{thm:bipartite-unequal} generalizes Theorem \ref{thm:star}. When the number of spokes in the star $K_{1, n}$ is odd (i.e., $n-1$ is even), there are two possible houses that can be allocated to the center in an optimal allocation. However, when the number of spokes is even (i.e., $n-1$ is odd), any optimal allocation allocates a unique house to the central node.

\subsection{Rooted Binary Trees}\label{sec:binary-trees}
In this section, we consider binary trees. A {\em binary tree} $T$ is defined as a rooted tree where each node has either $0$ or $2$ children. 

Our main result is a structural property characterizing at least one of the optimal allocations for any instance where the graph $G$ corresponds to a binary tree. 
We call this the {\em local median} property.
\begin{definition}[Local Median Property]
An allocation on a binary tree satisfies the local median property if, for any internal node, exactly one of its children is allocated a house with value less than that of the node. In other words, the value allocated to any internal node is the median of the set containing the node and its children.
\end{definition}


The proofs in this section will use the following lemma. We define the {\em inverse} of a valuation function $v$ as a valuation function $v^{inv}$ such that $v^{inv} (h) = - v(h)$ for all $h \in H$ (appropriately shifted so that all values are nonnegative). We note that any allocation has the same envy along any edge with respect to the inverted valuation and the original valuation, whose straightforward proof we relegate to Appendix \ref{apdx:connected}.

\begin{restatable}{lemma}{lemvaluationinverse}\label{lem:valuation-inverse}
The envy along any edge of the graph $G$ under an allocation $\pi$ with respect to the valuation $v$ is equal to the envy along the same edge of the graph $G$ under the allocation $\pi$ with respect to the valuation $v^{inv}$.
\end{restatable}


We will now show that at least one minimum envy allocation satisfies the local median property. 
More formally, we show the following: Given a binary tree $T$ and any allocation $\pi$, there exists an allocation that satisfies the local median property and has equal or lower total envy. 
The proof relies on the following lemma.


\begin{restatable}{lemma}{lemlocalmedianimprovement}\label{lem:local-median-improvement}
Let $\pi$ be an allocation on a binary tree $T$, not satisfying the local median property. Let $i$ be the internal node furthest from the root which is not allocated the median among the values given to it and its children. Then, there exists an allocation $\pi'$ such that 
\begin{enumerate}[(a)]
    \item For the subtree $T'$ rooted at $i$, we have that $\Envy(\pi(T'), T') > \Envy(\pi'(T'), T')$;
    \item For any other subtree $T''$ not contained by $T'$, we have that $\Envy(\pi(T''), T'') \ge \Envy(\pi'(T''), T'')$.
\end{enumerate}
\end{restatable}

\begin{proof}[Proof Sketch]
Consider an internal vertex not satisfying the local median property; by Lemma \ref{lem:valuation-inverse}, it suffices to consider the case when this violating vertex has a value that is less than both its children. We can now ``push'' the value allocated to this vertex down the tree by a series of swaps until that value reaches a leaf or satisfies the local median property. It can be shown that the total envy at the end satisfies the two criteria stated. A full proof is given in Appendix \ref{apdx:connected}.
\end{proof}

Lemma \ref{lem:local-median-improvement} immediately gives rise to the following corollary, which we state as a theorem.

\begin{restatable}{theorem}{thmlocalmedian}\label{thm:localmedian}
For any binary tree $T$, at least one minimum envy allocation satisfies the local median property.
\end{restatable}

Unfortunately, the local median property is too weak to exploit for a polynomial-time algorithm. Ideally, we would like to use the property to show that some minimum envy allocation satisfies an even stronger property called the \emph{global median} property.

\begin{definition}[Global Median Property]
An allocation on a binary tree satisfies the global median property if, for every internal node, all the houses in one subtree of the node have value less than the house allocated to the node, and all the houses in the other subtree have value greater than the house allocated to the node.
\end{definition}

Empirically, it seems like minimum envy allocations on rooted binary trees do indeed satisfy the global median property. If this turned out to be true, we would be able to devise an algorithm that significantly reduces the search space for a minimum envy allocation from $\Theta(n!)$ to $\Theta(2^d)$. This would give us an $O(2^d)$ algorithm for computing minimum envy allocations on a rooted binary tree of maximum depth $d$.

\begin{conjecture}
There is an algorithm that computes an optimal house allocation on a rooted binary tree of maximum depth $d$ in time $O(2^d)$. In particular, this algorithm runs in polynomial time on (nearly) balanced trees.
\end{conjecture}





\subsection{General Trees}\label{sec:general-trees}

How do we take the approaches for rooted binary trees and build towards arbitrary trees? Note that one consequence of Theorem \ref{thm:localmedian} is that in an optimal allocation on a rooted binary tree, the minimum and the maximum must appear on leaves.

In the minimum linear arrangement problem, it is known~\citep{seidvasser} that when the underlying graph is a tree, some optimal allocation assigns the minimum and maximum values to leaves, and furthermore, the (unique) path from this minimum to the maximum consists of monotonically increasing values. This characterization is used crucially in designing the polynomial time algorithm on trees in this context~\citep{mlatrees}.

Empirically, this same property for trees seems to hold for non-uniformly spaced values as well. The proof technique used in~\citep{seidvasser} does not extend to our setting, but testing the problem on 200 randomly generated trees and uniformly random values on the interval $[0, 100]$ always gave us both these properties on trees: the minimum and maximum values both end up on leaves, with a monotonic path between them.

We believe that graphical house allocation, unlike minimum linear arrangements, is NP-hard on trees. It would be remarkable if the structural characterization holds for our problem, but the polynomial-time algorithm does not work. We relegate answering this to future work.

\section{Disconnected Graphs}\label{sec:disconnected}

In this section, we consider disconnected graphs, starting with a structural characterization, and then using that to obtain upper bounds for several natural classes of disconnected graphs that have lower bounds in Section \ref{sec:lowerbounds}.

\subsection{A Structural Characterization}

We start by remarking that Proposition \ref{prop:mlawlog} is \emph{false} in our setting, and so we can no longer assume our graph is connected without loss of generality. For instance, consider an instance when the underlying graph $G$ is a disjoint union of an edge and a triangle. The two valuation intervals in Figure \ref{fig:disjointpathcyclecounter} yield very different optimal structures for this same instance.
\begin{figure}
    \centering
    \small
    \begin{tikzpicture}
    \xdef\eps{0.05}
        \foreach \x in {1,2,14,15,16}
            {
                \coordinate (A\x) at ($(\x/2,0)$) {};
                \draw ($(A\x)+(0,3pt)$) -- ($(A\x)-(0,3pt)$);
            }
            \draw[line width=0.5 mm] (0.5,0) -- (8,0);
            \draw[color=thechosenone,line width=0.5 mm] (0.5+\eps,-0.25) -- (1-\eps,-0.25);
            \draw[color=thechosenone,line width=0.5 mm] (7+\eps,-0.25) -- (7.5-\eps,-0.25);
            \draw[color=thechosenone,line width=0.5 mm] (7.5+\eps,-0.25) -- (8-\eps,-0.25);
            \draw[color=thechosenone,line width=0.5 mm] (7+\eps,-0.5) -- (8-\eps,-0.5);
        \foreach \x in {1,7,8,9,16}
            {
                \coordinate (A\x) at ($(\x/2,-1.5)$) {};
                \draw ($(A\x)+(0,3pt)$) -- ($(A\x)-(0,3pt)$);
            }
            \draw[line width=0.5 mm] (0.5,-1.5) -- (8,-1.5);
            \draw[color=thechosenone,line width=0.5 mm] (0.5+\eps,-1.75) -- (8-\eps,-1.75);
            \draw[color=thechosenone,line width=0.5 mm] (3.5+\eps,-2) -- (4-\eps,-2);
            \draw[color=thechosenone,line width=0.5 mm] (4+\eps,-2) -- (4.5-\eps,-2);
            \draw[color=thechosenone,line width=0.5 mm] (3.5+\eps,-2.25) -- (4.5-\eps,-2.25);
    \end{tikzpicture}
    \caption{For the valuation interval on top, the optimal allocation to $P_2 + C_3$ is to give the two low-valued houses to the edge, and to give the three high-valued houses to the triangle. This is the only allocation where the envy is negligible. For the valuation interval on the bottom, the optimal allocation to $P_2 + C_3$ is to give the two extreme-valued houses to the edge, and the cluster in the middle to the triangle. Any other allocation has to count one of the long halves of the interval multiple times, and is therefore strictly suboptimal. This is an instance where we see one of the connected components being ``split'' by another in the valuation interval.}
    \label{fig:disjointpathcyclecounter}
\end{figure}
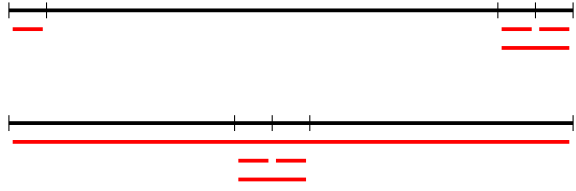

We remark that this is a major departure from the linear arrangement problem, as the spacing of the values along the valuation interval becomes a key factor in the structure of optimal allocations. Recall from Proposition \ref{prop:mlawlog} that in the minimum linear arrangement problem, in an optimal allocation, the connected components of a graph take the houses in contiguous subsets along the valuation interval. The example in Figure \ref{fig:disjointpathcyclecounter} shows that this is not necessarily true in our problem, and therefore serves as a motivation to classify disconnected graphs according to whether they have this property or not. We call the relevant property \emph{separability}, defined as follows.


\begin{restatable}[Splitting]{definition}{defcontiguous}\label{def:contiguous}
Let $G_1 = (N_1, E_1)$ and $G_2 = (N_2, E_2)$ be two of the connected components of $G = (N, E)$, and fix an arbitrary allocation $\pi$. For $i \in \{1, 2\}$, let $\pi(G_i)$ be the set of values given to vertices in $G_i$. We say \emph{$G_1$ splits $G_2$ in $\pi$} if the values of $\pi(G_1)$ form a contiguous subset of the values in $\pi(G_1) \cup \pi(G_2)$.
\end{restatable}

\begin{restatable}[Separability]{definition}{defseparable}\label{def:separable}
Let $G$ be a disconnected graph with connected components $G_1, \ldots, G_k$. Then,
\begin{enumerate}
    \item $G$ is \emph{separable} if there exists an ordering $G_1, \ldots, G_k$ of the components where, for all valuation intervals, there is an optimal allocation where for all $1 \leq i < j \leq k$, $G_i$ splits $G_j$.
    \item $G$ is \emph{strongly separable} if, in addition, $G_j$ also splits $G_i$. Note that this is only possible if some optimal allocation assigns contiguous subsets of values to all connected components.
    \item $G$ is \emph{inseparable} if there is a valuation interval where for each optimal allocation $\pi$, there are components $G_1$ and $G_2$ with $u, u' \in \pi(G_1)$ and $v, v' \in \pi(G_2)$ such that $u < v < u' < v'$.
\end{enumerate}
A class $\mathcal{A}$ of graphs is separable (resp. strongly separable) if every graph in it is separable (resp. strongly separable). Conversely, $\mathcal{A}$ is inseparable if it contains an inseparable graph.
\end{restatable}

We note that a graph $G$ is inseparable if and only if it is not separable, 
and furthermore, it is strongly separable only if it is separable. 

For minimum linear arrangements, all disconnected graphs are strongly separable, by Proposition \ref{prop:mlawlog}. In contrast, for our problem, Figure~\ref{fig:disjointpathcyclecounter} already provides an example of a graph that is not separable (equivalently, inseparable). We discuss several examples of strongly separable graphs in our problem in Section~\ref{subsec:Disjoint-Unions-Paths-Cycles-Stars}; in particular, disjoint unions of paths (respectively, cycles or stars) satisfy strong separability.

Our formulation of separability and strong separability has an immediate algorithmic consequence. Specifically, if $G$ is strongly separable and has $k$ connected components, if we have an efficient algorithm on each connected component individually, then we have an FPT algorithm on $G$, parameterized by $k$. Similarly, if $G$ is separable, if we have an efficient algorithm on each connected component, then we have an XP algorithm on $G$. 

It is not immediately obvious that there are separable graphs that are not strongly separable. 
Figure \ref{fig:disjointpathcyclecounter} shows an example of such a graph (Theorem~\ref{thm:disjointcliques} proves separability). We will see more examples of this later, but we remark that there are even separable forests that are not strongly separable (Figure \ref{fig:spannerpoint}). We state this formally here, relegating the proof once again to Appendix \ref{apdx:disconnected}.

\begin{restatable}{proposition}{propseparablenotstrong}\label{propseparablenotstrong}
There exists a separable forest that is not strongly separable.
\end{restatable}

\begin{figure}[ht]
    \centering
    %
    \begin{tikzpicture}
        \tikzset{mynode/.style = {shape=circle,draw,inner sep=1.5pt}}
        \tikzset{edge/.style = {solid}}
        \node[mynode] (1) at (0,0) {};
        \node[mynode] (2) at (-0.4,0.5) {};
        \node[mynode] (3) at (-0.4,-0.5) {};
        \node[mynode] (4) at (2,0) {};
        \node[mynode] (5) at (2.4,0.5) {};
        \node[mynode] (6) at (2.4,-0.5) {};
        \node[mynode] (7) at (4,0) {};
        \draw (1) -> (2);
        \draw (1) -> (3);
        \draw (1) -> (4);
        \draw (4) -> (5);
        \draw (4) -> (6);
        \draw ($(0,-1.5)+(0,3pt)$) -- ($(0,-1.5)-(0,3pt)$);
        \draw ($(0.1,-1.5)+(0,3pt)$) -- ($(0.1,-1.5)-(0,3pt)$);
        \draw ($(0.2,-1.5)+(0,3pt)$) -- ($(0.2,-1.5)-(0,3pt)$);
        \draw ($(0.3,-1.5)+(0,3pt)$) -- ($(0.3,-1.5)-(0,3pt)$);
        \draw ($(0.4,-1.5)+(0,3pt)$) -- ($(0.4,-1.5)-(0,3pt)$);
        \draw ($(0.5,-1.5)+(0,3pt)$) -- ($(0.5,-1.5)-(0,3pt)$);
        \draw ($(4,-1.5)+(0,3pt)$) -- ($(4,-1.5)-(0,3pt)$);
        \draw[line width=0.5 mm] (0,-1.5) -- (4,-1.5);
        %
        \draw[color=thechosenone,line width=0.5 mm] (0,-1.95) -- (0.2,-1.95);
        \draw[color=thechosenone,line width=0.5 mm] (0.1,-1.85) -- (0.2,-1.85);
        \draw[color=thechosenone,line width=0.5 mm] (0.2,-1.75) -- (0.3,-1.75);
        \draw[color=thechosenone,line width=0.5 mm] (0.3,-1.95) -- (0.5,-1.95);
        \draw[color=thechosenone,line width=0.5 mm] (0.3,-1.85) -- (0.4,-1.85);
        \draw[color=thechosenone,line width=0.5 mm] (3.95,-1.75) -- (4.05,-1.75);
        \draw ($(0,-2.5)+(0,3pt)$) -- ($(0,-2.5)-(0,3pt)$);
        \draw ($(0.1,-2.5)+(0,3pt)$) -- ($(0.1,-2.5)-(0,3pt)$);
        \draw ($(0.2,-2.5)+(0,3pt)$) -- ($(0.2,-2.5)-(0,3pt)$);
        \draw ($(2,-2.5)+(0,3pt)$) -- ($(2,-2.5)-(0,3pt)$);
        \draw ($(3.8,-2.5)+(0,3pt)$) -- ($(3.8,-2.5)-(0,3pt)$);
        \draw ($(3.9,-2.5)+(0,3pt)$) -- ($(3.9,-2.5)-(0,3pt)$);
        \draw ($(4,-2.5)+(0,3pt)$) -- ($(4,-2.5)-(0,3pt)$);
        \draw[line width=0.5 mm] (0,-2.5) -- (4,-2.5);
        %
        \draw[color=thechosenone,line width=0.5 mm] (0.1,-2.75) -- (0.2,-2.75);
        \draw[color=thechosenone,line width=0.5 mm] (0,-2.85) -- (0.2,-2.85);
        \draw[color=thechosenone,line width=0.5 mm] (1.95,-2.75) -- (2.05,-2.75);
        \draw[color=thechosenone,line width=0.5 mm] (3.8,-2.75) -- (3.9,-2.75);
        \draw[color=thechosenone,line width=0.5 mm] (3.8,-2.85) -- (4,-2.85);
        \draw[color=thechosenone,line width=0.5 mm] (0.2,-2.95) -- (3.8,-2.95);
    \end{tikzpicture}
    \caption{Example of a separable forest that is not strongly separable. The forest is trivially separable. For the bottom valuation line, an optimal allocation must allocate the extreme clusters in the interval to the larger connected component. See Appendix \ref{apdx:disconnected} for a detailed proof.}
    \label{fig:spannerpoint}
\end{figure}
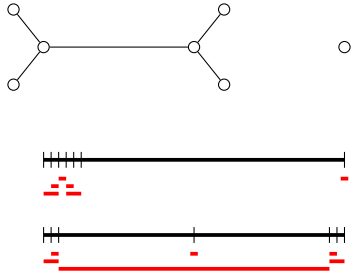

Less obviously, \emph{inseparable forests} exist, as shown by the following proposition (and Figure \ref{fig:degreeseparated}), whose proof is in Appendix \ref{apdx:disconnected}.

\begin{restatable}{proposition}{propinseparable}\label{propinseparable}
There exists an inseparable forest.
\end{restatable}

\begin{figure}[ht]
    \centering
    %
    \begin{tikzpicture}
        \tikzset{mynode/.style = {shape=circle,draw,inner sep=1.5pt}}
        \tikzset{edge/.style = {solid}}
        \node[mynode] (1) at (0,0) {};
        \node[mynode] (2) at (-0.4,0.5) {};
        \node[mynode] (3) at (-0.5,0.2) {};
        \node (0) at (-0.4,0) {$\vdots$};
        \node (0) at (-0.4,0.8) {$s_1$};
        \node[mynode] (4) at (-0.4,-0.5) {};
        \node[mynode] (5) at (2,0) {};
        \node[mynode] (6) at (2.4,0.5) {};
        \node[mynode] (7) at (2.5,0.2) {};
        \node (0) at (2.4,0) {$\vdots$};
        \node (0) at (2.4,0.8) {$s_3$};
        \node[mynode] (8) at (2.4,-0.5) {};
        \draw (1) -> (2);
        \draw (1) -> (3);
        \draw (1) -> (4);
        \draw (5) -> (6);
        \draw (5) -> (7);
        \draw (5) -> (8);
        \draw (1) -> (5);
        %
        %
        \node[mynode] (1) at (4,0) {};
        \node[mynode] (2) at (3.6,0.5) {};
        \node[mynode] (3) at (3.5,0.2) {};
        \node (0) at (3.6,0) {$\vdots$};
        \node (0) at (3.6,0.8) {$s_2$};
        \node[mynode] (4) at (3.6,-0.5) {};
        \node[mynode] (5) at (6,0) {};
        \node[mynode] (6) at (6.4,0.5) {};
        \node[mynode] (7) at (6.5,0.2) {};
        \node (0) at (6.4,0) {$\vdots$};
        \node (0) at (6.4,0.8) {$s_4$};
        \node[mynode] (8) at (6.4,-0.5) {};
        \draw (1) -> (2);
        \draw (1) -> (3);
        \draw (1) -> (4);
        \draw (5) -> (6);
        \draw (5) -> (7);
        \draw (5) -> (8);
        \draw (1) -> (5);
        \draw ($(0,-1.5)+(0,3pt)$) -- ($(0,-1.5)-(0,3pt)$);
        \draw ($(0.05,-1.5)+(0,3pt)$) -- ($(0.05,-1.5)-(0,3pt)$);
        \draw ($(0.1,-1.5)+(0,3pt)$) -- ($(0.1,-1.5)-(0,3pt)$);
        \draw ($(0.15,-1.5)+(0,3pt)$) -- ($(0.15,-1.5)-(0,3pt)$);
        \node (0) at (0.1,-1.8) {$s_1 + 1$};
        \draw ($(2,-1.5)+(0,3pt)$) -- ($(2,-1.5)-(0,3pt)$);
        \draw ($(2.05,-1.5)+(0,3pt)$) -- ($(2.05,-1.5)-(0,3pt)$);
        \draw ($(2.1,-1.5)+(0,3pt)$) -- ($(2.1,-1.5)-(0,3pt)$);
        \draw ($(2.15,-1.5)+(0,3pt)$) -- ($(2.15,-1.5)-(0,3pt)$);
        \node (0) at (2.1,-1.8) {$s_2 + 1$};
        \draw ($(4,-1.5)+(0,3pt)$) -- ($(4,-1.5)-(0,3pt)$);
        \draw ($(4.05,-1.5)+(0,3pt)$) -- ($(4.05,-1.5)-(0,3pt)$);
        \draw ($(4.1,-1.5)+(0,3pt)$) -- ($(4.1,-1.5)-(0,3pt)$);
        \draw ($(4.15,-1.5)+(0,3pt)$) -- ($(4.15,-1.5)-(0,3pt)$);
        \node (0) at (4.1,-1.8) {$s_3 + 1$};
        \draw ($(6,-1.5)+(0,3pt)$) -- ($(6,-1.5)-(0,3pt)$);
        \draw ($(6.05,-1.5)+(0,3pt)$) -- ($(6.05,-1.5)-(0,3pt)$);
        \draw ($(6.1,-1.5)+(0,3pt)$) -- ($(6.1,-1.5)-(0,3pt)$);
        \draw ($(6.15,-1.5)+(0,3pt)$) -- ($(6.15,-1.5)-(0,3pt)$);
        \node (0) at (6.1,-1.8) {$s_4 + 1$};
        \draw[line width=0.5 mm] (0,-1.5) -- (6.15,-1.5);
        %
        \draw[color=thechosenone,line width=0.5 mm] (0.15,-2.5) -- (4,-2.5);
        \draw[color=thechosenone,line width=0.5 mm] (0.1,-2.35) -- (0.15,-2.35);
        \draw[color=thechosenone,line width=0.5 mm] (0.05,-2.25) -- (0.15,-2.25);
        \draw[color=thechosenone,line width=0.5 mm] (0,-2.15) -- (0.15,-2.15);
        \draw[color=thechosenone,line width=0.5 mm] (4,-2.35) -- (4.05,-2.35);
        \draw[color=thechosenone,line width=0.5 mm] (4,-2.25) -- (4.1,-2.25);
        \draw[color=thechosenone,line width=0.5 mm] (4,-2.15) -- (4.15,-2.15);
        \draw[color=thechosenone,line width=0.5 mm] (2,-2.65) -- (2.15,-2.65);
        \draw[color=thechosenone,line width=0.5 mm] (2.05,-2.75) -- (2.15,-2.75);
        \draw[color=thechosenone,line width=0.5 mm] (2.1,-2.85) -- (2.15,-2.85);
        \draw[color=thechosenone,line width=0.5 mm] (6,-2.65) -- (6.15,-2.65);
        \draw[color=thechosenone,line width=0.5 mm] (6,-2.75) -- (6.1,-2.75);
        \draw[color=thechosenone,line width=0.5 mm] (6,-2.85) -- (6.05,-2.85);
        \draw[color=thechosenone,line width=0.5 mm] (2.15,-3) -- (6,-3);
    \end{tikzpicture}
    \caption{Example of an inseparable forest. Suppose $s_1 < s_2 < s_3 < s_4$, and they satisfy for all $i, j$, $|s_i - s_j| \geq 3$, and for all $i, j, k$, $s_i + s_j > s_k + 2$. Then, an optimal allocation on this instance must allocate the entire cluster of size $s_i + 1$ on the valuation interval to the corresponding star-like cluster of the given forest. See Appendix \ref{apdx:disconnected} for a detailed proof.}
    \label{fig:degreeseparated}
\end{figure}
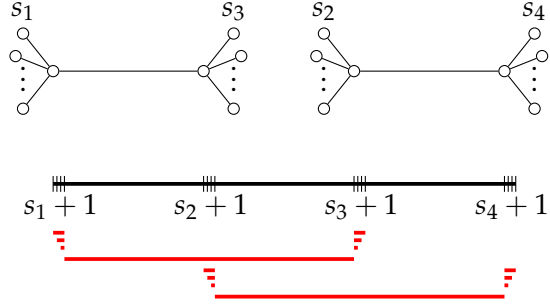

\subsection{Disjoint Unions of Paths, Cycles, and Stars}
\label{subsec:Disjoint-Unions-Paths-Cycles-Stars}

We now move on to algorithmic approaches and characterizations of minimum envy allocations, and start with the setting where $G$ is a disjoint union of paths. Suppose $G = P_{n_1} + \ldots + P_{n_r}$. What does an optimal 
allocation on $G$ look like?

\begin{restatable}{theorem}{thmdisjointpaths}\label{thm:disjointpaths}
Let $G$ be a disjoint union of paths, $P_{n_1} + \ldots + P_{n_r}$. Then, $G$ is strongly separable. Furthermore, in any optimal allocation, within each path, the houses appear in sorted order.
\end{restatable}

\begin{proof}[Proof Sketch]
In any allocation $\pi$, we may assume by Theorem \ref{thm:path} that the houses allocated to each path appear in sorted order along that path. Next, in an allocation $\pi$, if there is any overlap between two of the paths, then we can remove the overlap by reassigning the houses among just those two paths, and obtain an allocation with strictly less envy.
\end{proof}

The following corollary shows an FPT algorithm on the disjoint union of paths, parameterized by the number of different paths.

\begin{restatable}{corollary}{cordisjointpathsalgorithm}\label{cor:disjointpathsalgorithm}
We can find an optimal house allocation for an instance on an undirected $n$-agent graph $G$ that is the disjoint union of paths in time $O(nr!)$, where $r$ is the number of paths. 
\end{restatable}

\begin{proof}[Proof Sketch]
By Theorem \ref{thm:disjointpaths}, it suffices to find the optimal ordering of the paths. There are $r!$ such orderings, and for each, we can test the envy in linear time.
\end{proof}

If $G$ is a disjoint union of cycles, say $G = C_{n_1} + \ldots + C_{n_r}$, the same theorems characterizing optimal allocations go through, using Theorem \ref{thm:cycle}. We omit the proofs, but state the results formally.

\begin{restatable}{theorem}{thmdisjointcycles}\label{thm:disjointcycles}
Let $G$ be a disjoint union of cycles, $C_{n_1} + \ldots + C_{n_r}$. Then $G$ is strongly separable. Furthermore, in any optimal allocation, within each cycle, the houses appear in the form characterized in Theorem \ref{thm:cycle}.
\end{restatable}

\begin{restatable}{corollary}{cordisjointcyclesalgorithm}\label{cor:disjointcyclesalgorithm}
We can find an optimal house allocation for an instance on an undirected $n$-agent graph $G$ that is the disjoint union of cycles in time $O(nr!)$, where $r$ is the number of cycles.
\end{restatable}

We remark here that if $t$ is the number of different path (or cycle) \emph{lengths}, then a dynamic programming algorithm (Proposition \ref{prop:distinctpathlengths}) computes the minimum envy allocation in time $O(n^{t+1})$. Therefore, combining the two approaches, we have a time complexity of $O(\min(nr!, n^{t+1}))$. An immediate application of this dynamic programming algorithm is for graphs with degree at most one. These graphs are special case of the disjoint union of paths where the path length can either be $0$ or $1$. 

\begin{restatable}{corollary}{cordegreeonealgorithm}\label{cor:degree1algorithm}
We can find an optimal house allocation for an instance on an undirected $n$-agent graph $G$ with maximum degree $1$ in time $O(n^3)$.

\end{restatable}

Perhaps remarkably, there is no particularly elegant characterization when the underlying graph $G$ is a disjoint union of paths \emph{and} cycles, even when there is only one path and one cycle. This is a consequence of Figure \ref{fig:disjointpathcyclecounter}.

Finally, a similar result holds for disjoint unions of stars, though the proof is somewhat different. We omit the proof of Corollary~\ref{cor:disjointstarsalgorithm}, which follows from Theorem~\ref{thm:disjointstars}.

\begin{restatable}{theorem}{thmdisjointstars}\label{thm:disjointstars}
Let $G$ be a disjoint union of stars, $K_{1, n_1} + \ldots + K_{1, n_r}$. Then $G$ is strongly separable. Furthermore, in any optimal allocation, within each star, the houses appear in the form characterized in Theorem \ref{thm:star}.
\end{restatable}

\begin{proof}[Proof Sketch]
We can ``separate'' any two stars while improving  our objective. Given any allocation $\pi$, suppose one of the stars $K_{1, n_1}$ splits another star $K_{1, n_2}$ (but not vice-versa). Suppose their centers receive houses $h_i \prec h_j$ respectively. We show that exchanging $K_{1, n_1}$'s most-valuable house for $K_{1, n_2}$'s least-valuable house strictly decreases envy. We apply this repeatedly until the two stars split each other. We apply Theorem \ref{thm:star} to every star at the beginning of the process and after every swap to maintain the invariant that each star has a median house at its center.
\end{proof}

\begin{restatable}{corollary}{cordisjointstarsalgorithm}\label{cor:disjointstarsalgorithm}
We can find an optimal house allocation for an instance on an undirected $n$-agent graph $G$ that is the disjoint union of stars in time $O(nr!)$, where $r$ is the number of stars.
\end{restatable}

\subsection{Disjoint Unions of Cliques}

We now turn our attention to disjoint unions of cliques. We first demonstrate that when all cliques have the same size, we maintain strong separability. 

\begin{restatable}{theorem}{thmdisjointequicliques}\label{thm:disjointequicliques}
Let $G$ be a disjoint union of cliques with equal sizes, $K^1_{n/r} + \ldots + K^r_{n/r}$. Then, $G$ is strongly separable.
\end{restatable}

\begin{proof}[Proof Sketch]
WLOG consider two cliques, say $K$ and $K'$, on $n/2$ vertices each, and consider an arbitrary instance with $n$ values. Suppose in some allocation $\pi$, $K$ receives $A \cup A'$ and $K'$ receives $B \cup B'$, where $A \cup B$ form the $n/2$ lower-valued houses, and $A' \cup B'$ form the $n/2$ higher-valued houses. We can show that we improve the envy by assigning them $A \cup B$ and $A' \cup B'$ respectively.
\end{proof}

Because the cliques are all of equal sizes and agents have identical valuations, Theorem \ref{thm:disjointequicliques} implies that there is a trivial algorithm for assigning houses to agents. We can assign the first $n/r$ houses to one clique, the next $n/r$ houses to the next clique, and so on.

\begin{restatable}{corollary}{cordisjointequicliquesalgorithm}\label{cor:disjointequicliquesalgorithm}
We can find an optimal house allocation for an instance on an undirected $n$-agent graph $G$ that is the disjoint union of equal-sized cliques in time $O(n)$.
\end{restatable}

We now turn our attention to the case when the cliques are not all of the same size.

As we saw in Figure \ref{fig:disjointpathcyclecounter}, strong separability must be ruled out when cliques have different sizes. We will show that separability still holds. We show further that the largest clique gets a contiguous subset, and then subject to that, the second largest clique gets a contiguous subset, and so on. The proof is heavily technical, and we relegate the casework and details to Appendix \ref{apdx:disconnected}.

\begin{restatable}{theorem}{thmdisjointcliques}\label{thm:disjointcliques}
Let $G$ be a disjoint union of cliques with arbitrary sizes, $K_{n_1} + \ldots + K_{n_r}$. Then, $G$ is separable (but not strongly separable if the $n_i$'s are not all the same). In particular, if $n_1 \geq \ldots \geq n_r$, the ordering for separability is $K_{n_1}, \ldots, K_{n_r}$.
\end{restatable}

\begin{proof}[Proof Sketch]
WLOG consider two cliques $K$ and $K'$ with $|K| > |K'|$. We will show that $K$ receives a contiguous subset. The basic approach is that if there are three values $v(h_i) < v(h_j) < v(h_k)$ with $h_j$ going to $K'$ but $h_i, h_k$ going to $K$, we can swap houses around and obtain a better allocation by a counting argument. 
\end{proof}

Theorem \ref{thm:disjointcliques} implies an XP algorithm for finding a minimum envy allocation on unions of cliques.

\begin{restatable}{corollary}{cordisjointcliquesalgorithm}\label{cor:disjointcliquesalgorithm}
We can find an optimal house allocation for an instance on an undirected $n$-agent graph $G$ that is the disjoint union of cliques in time $O(n^{r+2})$, where $r$ is the number of cliques.
\end{restatable}

There seems to be a separation between unions of differently-sized cliques and unions of stars, cycles, and paths. We suspect the problem may be (at least)
W[1]-hard for unions of arbitrary cliques.

\section{Conclusion and Discussions}\label{sec:conclusion}

We investigate a generalization of the classical house allocation problem where the agents are on the vertices of a graph representing the underlying social network, and we wish to allocate the houses to the agents so as to minimize the aggregate envy among neighbors. Even for identical valuations, we show that the problem is computationally hard and structurally rich. Furthermore, our structural insights facilitate algorithmic results for several well-motivated graph classes.

A natural question for future research is to consider other fairness objectives such as \emph{minimizing the maximum envy} present on any edge of the graph. This corresponds to the classical graph theoretic property of \emph{bandwidth}, which is also known to be NP-complete for general graphs, and quite hard to approximate as well~\citep{P76np,BKW97approximation}. These questions are tied with the choice of which \emph{norm} to optimize our problem on. Indeed, optimizing for the total envy (as in our work) and the maximum envy (as in the bandwidth problem) correspond to optimizing over the $L_1$ norm and $L_\infty$ norm, respectively. Exploring these and other norms offers scope for future research.

Resolving the separation between the minimum linear arrangement problem and the graphical house allocation problem in the context of \emph{trees} is an important question. It would be interesting to know whether trees admit polynomial time characterizations of the minimum envy, or---more remarkably---whether they are NP-complete but admit the structural similarities to the minimum linear arrangement problem discussed in Section \ref{sec:general-trees}.




More generally, it is an interesting problem to find other natural classes of strongly separable graphs, and eventually to \emph{completely characterize} all strongly separable graphs in terms of their graph theoretic structure, for algorithmic approaches. Another important future direction would be to extend some of these results for \emph{non-identical} valuations.

\section*{Acknowledgments}
We would like to thank Andrew McGregor for pointing us to the linear arrangement problem, and for suggesting multiple useful lines of inquiry. We would also like to thank Cameron Musco for extremely helpful discussions. Finally, we would like to thank Yair Zick for being involved in many of the early discussions, and providing guidance throughout. RS acknowledges support from NSF grant no. CCF-1908849. RV acknowledges support from DST INSPIRE grant no. DST/INSPIRE/04/2020/000107.
HH acknowledges support from NSF IIS grants \#2144413 and \#2107173.

\bibliographystyle{plainnat} 
\bibliography{abb, references}

\begin{thebibliography}{32}
\providecommand{\natexlab}[1]{#1}
\providecommand{\url}[1]{\texttt{#1}}
\expandafter\ifx\csname urlstyle\endcsname\relax
  \providecommand{\doi}[1]{doi: #1}\else
  \providecommand{\doi}{doi: \begingroup \urlstyle{rm}\Url}\fi

\bibitem[Abdulkadiro{\u{g}}lu and S{\"o}nmez(1999)]{AS99house}
Atila Abdulkadiro{\u{g}}lu and Tayfun S{\"o}nmez.
\newblock {House Allocation with Existing Tenants}.
\newblock \emph{Journal of Economic Theory}, 88\penalty0 (2):\penalty0
  233--260, 1999.

\bibitem[Abebe et~al.(2017)Abebe, Kleinberg, and Parkes]{abebe2017fair}
Rediet Abebe, Jon Kleinberg, and David~C Parkes.
\newblock {Fair Division via Social Comparison}.
\newblock In \emph{Proceedings of the 16th International Conference on
  Autonomous Agents and MultiAgent Systems}, pages 281--289, 2017.

\bibitem[Ahuja et~al.(1988)Ahuja, Magnanti, and Orlin]{ahuja1988network}
Ravindra~K Ahuja, Thomas~L Magnanti, and James~B Orlin.
\newblock {Network Flows}.
\newblock 1988.

\bibitem[Aigner-Horev and Segal-Halevi(2022)]{aigner2022envy}
Elad Aigner-Horev and Erel Segal-Halevi.
\newblock {Envy-Free Matchings in Bipartite Graphs and their Applications to
  Fair Division}.
\newblock \emph{Information Sciences}, 587:\penalty0 164--187, 2022.

\bibitem[Beynier et~al.(2018{\natexlab{a}})Beynier, Chevaleyre, Gourv\`{e}s,
  Lesca, Maudet, and Wilczynski]{beynier2018localenvy}
Aur\'{e}lie Beynier, Yann Chevaleyre, Laurent Gourv\`{e}s, Julien Lesca,
  Nicolas Maudet, and Ana\"{e}lle Wilczynski.
\newblock {Local Envy-Freeness in House Allocation Problems}.
\newblock In \emph{Proceedings of the 17th International Conference on
  Autonomous Agents and MultiAgent Systems}, page 292–300,
  2018{\natexlab{a}}.

\bibitem[Beynier et~al.(2018{\natexlab{b}})Beynier, Maudet, and
  Damamme]{beynier2018fairness}
Aur{\'e}lie Beynier, Nicolas Maudet, and Anastasia Damamme.
\newblock {Fairness in Multiagent Resource Allocation with Dynamic and Partial
  Observations}.
\newblock In \emph{Proceedings of the 17th International Conference on
  Autonomous Agents and Multiagent Systems}, pages 1868--1870,
  2018{\natexlab{b}}.

\bibitem[Blache et~al.(1997)Blache, Karpi{\'n}ski, and
  Wirtgen]{BKW97approximation}
Gunter Blache, Marek Karpi{\'n}ski, and Juergen Wirtgen.
\newblock \emph{{On Approximation Intractability of the Bandwidth Problem}}.
\newblock Inst. f{\"u}r Informatik, 1997.

\bibitem[Bredereck et~al.(2022)Bredereck, Kaczmarczyk, and
  Niedermeier]{bredereck2022envy}
Robert Bredereck, Andrzej Kaczmarczyk, and Rolf Niedermeier.
\newblock {Envy-Free Allocations Respecting Social Networks}.
\newblock \emph{Artificial Intelligence}, 2022.

\bibitem[Bui and Jones(1992)]{bj92}
Thang~Nguyen Bui and Curt Jones.
\newblock {Finding Good Approximate Vertex and Edge Partitions is NP-Hard}.
\newblock \emph{Information Processing Letters}, 42\penalty0 (3):\penalty0
  153--159, 1992.
\newblock ISSN 0020-0190.

\bibitem[Chung(1984)]{mlatrees}
F.R.K. Chung.
\newblock {On Optimal Linear Arrangements of Trees}.
\newblock \emph{Computers \& Mathematics with Applications}, 10\penalty0
  (1):\penalty0 43--60, 1984.
\newblock ISSN 0898-1221.

\bibitem[Doucette et~al.(2019)Doucette, Tsang, Hosseini, Larson, and
  Cohen]{doucette2019inferring}
John~A Doucette, Alan Tsang, Hadi Hosseini, Kate Larson, and Robin Cohen.
\newblock {Inferring True Voting Outcomes in Homophilic Social Networks}.
\newblock \emph{Autonomous Agents and Multi-Agent Systems}, 33\penalty0
  (3):\penalty0 298--329, 2019.

\bibitem[Eiben et~al.(2020)Eiben, Ganian, Hamm, and
  Ordyniak]{eiben2020parameterized}
Eduard Eiben, Robert Ganian, Thekla Hamm, and Sebastian Ordyniak.
\newblock {Parameterized Complexity of Envy-Free Resource Allocation in Social
  Networks}.
\newblock In \emph{Proceedings of the Thirty-Fourth AAAI Conference on
  Artificial Intelligence}, pages 7135--7142, 2020.

\bibitem[Even and Shiloach(1978)]{mlabinaryhard}
S.~Even and Y.~Shiloach.
\newblock {NP-Completeness of Several Arrangements Problems}.
\newblock \emph{Technical Report, TR-43 The Technicon}, page~29, 1978.
\newblock URL
  \url{https://www.cs.technion.ac.il/users/wwwb/cgi-bin/tr-info.cgi/1975/CS/CS0043}.

\bibitem[Gan et~al.(2019)Gan, Suksompong, and Voudouris]{gsvfairhouse}
Jiarui Gan, Warut Suksompong, and Alexandros~A. Voudouris.
\newblock {Envy-Freeness in House Allocation Problems}.
\newblock \emph{Mathematical Social Sciences}, 101:\penalty0 104--106, 2019.
\newblock ISSN 0165-4896.

\bibitem[Garey et~al.(1976)Garey, Johnson, and Stockmeyer]{mlahard}
M.R. Garey, D.S. Johnson, and L.~Stockmeyer.
\newblock {Some Simplified NP-Complete Graph Problems}.
\newblock \emph{Theoretical Computer Science}, 1\penalty0 (3):\penalty0
  237--267, 1976.
\newblock ISSN 0304-3975.

\bibitem[Gourv{\`e}s et~al.(2017)Gourv{\`e}s, Lesca, and
  Wilczynski]{gourves2017object}
Laurent Gourv{\`e}s, Julien Lesca, and Ana{\"e}lle Wilczynski.
\newblock {Object Allocation via Swaps along a Social Network}.
\newblock In \emph{Proceedings of the 26th International Joint Conference on
  Artificial Intelligence}, pages 213--219, 2017.

\bibitem[Grandi(2017)]{grandi2017social}
Umberto Grandi.
\newblock {Social Choice and Social Networks}.
\newblock \emph{Trends in Computational Social Choice}, pages 169--184, 2017.

\bibitem[Igarashi and Elkind(2016)]{igarashi2016hedonic}
Ayumi Igarashi and Edith Elkind.
\newblock {Hedonic Games with Graph-Restricted Communication}.
\newblock In \emph{Proceedings of the 2016 International Conference on
  Autonomous Agents \& Multiagent Systems}, pages 242--250, 2016.

\bibitem[Ito et~al.(2022)Ito, Iwamasa, Kakimura, Kamiyama, Kobayashi, Nozaki,
  Okamoto, and Ozeki]{ito2022reforming}
Takehiro Ito, Yuni Iwamasa, Naonori Kakimura, Naoyuki Kamiyama, Yusuke
  Kobayashi, Yuta Nozaki, Yoshio Okamoto, and Kenta Ozeki.
\newblock {Reforming an Envy-Free Matching}.
\newblock In \emph{Proceedings of the Thirty-Sixth AAAI Conference on
  Artificial Intelligence}, pages 5084--5091, 2022.

\bibitem[Jansen et~al.(2013)Jansen, Kratsch, Marx, and
  Schlotter]{binpackinghardness}
Klaus Jansen, Stefan Kratsch, Dániel Marx, and Ildikó Schlotter.
\newblock {Bin Packing with Fixed Number of Bins Revisited}.
\newblock \emph{Journal of Computer and System Sciences}, 79\penalty0
  (1):\penalty0 39--49, 2013.
\newblock ISSN 0022-0000.

\bibitem[Kamiyama(2021)]{kamiyama2021envy}
Naoyuki Kamiyama.
\newblock {The Envy-Free Matching Problem with Pairwise Preferences}.
\newblock \emph{Information Processing Letters}, 172:\penalty0 106158, 2021.

\bibitem[Kamiyama et~al.(2021)Kamiyama, Manurangsi, and
  Suksompong]{kmsfairhouse}
Naoyuki Kamiyama, Pasin Manurangsi, and Warut Suksompong.
\newblock {On the Complexity of Fair House Allocation}.
\newblock \emph{Operations Research Letters}, 49\penalty0 (4):\penalty0
  572--577, 2021.
\newblock ISSN 0167-6377.

\bibitem[Khot(2006)]{khot04}
Subhash Khot.
\newblock {Ruling Out PTAS for Graph Min-Bisection, Dense k-Subgraph, and
  Bipartite Clique}.
\newblock \emph{SIAM Journal on Computing}, 36\penalty0 (4):\penalty0
  1025--1071, 2006.

\bibitem[Lange and Rothe(2019)]{lange2019optimizing}
Pascal Lange and J{\"o}rg Rothe.
\newblock {Optimizing Social Welfare in Social Networks}.
\newblock In \emph{Proceedings of the 6th International Conference on
  Algorithmic Decision Theory}, pages 81--96, 2019.

\bibitem[Massand and Simon(2019)]{graphicalonesided}
Sagar Massand and Sunil Simon.
\newblock {Graphical One-Sided Markets}.
\newblock In \emph{Proceedings of the Twenty-Eighth International Joint
  Conference on Artificial Intelligence}, pages 492--498, 2019.

\bibitem[Papadimitriou(1976)]{P76np}
Christos Papadimitriou.
\newblock {The NP-Completeness of the Bandwidth Minimization Problem}.
\newblock \emph{Computing}, 16\penalty0 (3):\penalty0 263--270, 1976.

\bibitem[Peters(2016)]{peters2016graphical}
Dominik Peters.
\newblock {Graphical Hedonic Games of Bounded Treewidth}.
\newblock In \emph{Proceedings of the Thirtieth AAAI Conference on Artificial
  Intelligence}, pages 586--593, 2016.

\bibitem[Roth et~al.(2004)Roth, S{\"o}nmez, and {\"U}nver]{roth2004kidney}
Alvin~E Roth, Tayfun S{\"o}nmez, and M~Utku {\"U}nver.
\newblock {Kidney Exchange}.
\newblock \emph{The Quarterly Journal of Economics}, 119\penalty0 (2):\penalty0
  457--488, 2004.

\bibitem[Seidvasser(1970)]{seidvasser}
M.~A. Seidvasser.
\newblock {The Optimal Number of the Vertices of a Tree}.
\newblock \emph{Diskref. Anal.}, 19:\penalty0 56--74, 1970.

\bibitem[Shapley and Scarf(1974)]{shapley1974cores}
Lloyd Shapley and Herbert Scarf.
\newblock {On Cores and Indivisibility}.
\newblock \emph{Journal of Mathematical Economics}, 1\penalty0 (1):\penalty0
  23--37, 1974.

\bibitem[Svensson(1999)]{svensson1999strategy}
Lars-Gunnar Svensson.
\newblock {Strategy-Proof Allocation of Indivisible Goods}.
\newblock \emph{Social Choice and Welfare}, 16\penalty0 (4):\penalty0 557--567,
  1999.

\bibitem[Tsang et~al.(2015)Tsang, Doucette, and Hosseini]{tsang2015voting}
Alan Tsang, John~A Doucette, and Hadi Hosseini.
\newblock {Voting with Social Influence: Using Arguments to Uncover Ground
  Truth}.
\newblock In \emph{Proceedings of the 14th International Conference on
  Autonomous Agents and Multiagent Systems}, pages 1841--1842, 2015.

\end{thebibliography}

\appendix
\newpage
\onecolumn

\section*{Appendix}

\section{Proofs from Section \ref{sec:prelims}}\label{apdx:prelims}

\begin{lemma}\label{lem:valspositive}
We may assume, without loss of generality, that the valuation function is one-to-one, i.e., for $h_1, h_2 \in H$, $h_1 \neq h_2 \implies v(h_1) \neq v(h_2)$.
\end{lemma}

\begin{proof}
Fix an arbitrary valuation function $v$, and set $\varepsilon > 0$ to be smaller than the minimum non-zero difference in envy between two allocations under $v$. We will show that there is a one-to-one valuation function $v'$ such that for any graph $G$ and any allocation $\pi$, the total envy under $v'$ differs from the total envy under $v$ by at most an additive term of $\varepsilon$. For $h_k \in H$, define
\begin{equation*}
    v'(h_k) = v(h_k) + \frac{\varepsilon}{n^22^k}.
\end{equation*}
Then, for any allocation $\pi$ on $G$, consider the envy between agents $i$ and $j$. If $\pi(i) = h_k$ and $\pi(j) = h_\ell$, we have, using the triangle inequality,
\begin{align*}
    \big|v'(\pi(i)) - v'(\pi(j))\big| &= \bigg|v(\pi(i)) - v(\pi(j)) + \frac{\varepsilon}{n^22^k} - \frac{\varepsilon}{n^22^\ell}\bigg| \\
    &\leq \big|v(\pi(i)) - v(\pi(j))\big| + \frac{\varepsilon}{n^2}\bigg|\frac{1}{2^k} - \frac{1}{2^\ell}\bigg| \\
    &< \big|v(\pi(i)) - v(\pi(j))\big| + \frac{\varepsilon}{n^2}.
\end{align*}
Summing over the at most $n^2$ edges of $G$, we have $\Envy_{v'}(G, \pi) < \Envy_v(G, \pi) + \varepsilon$, as desired, where the subscripts $v$ and $v'$ denote the valuation functions being used in each case.

For any solution $\pi^{\ast}$ which minimizes envy under $v'$, if we compare against another solution $\pi'$, we see that 

\begin{equation*}
\Envy_v(G, \pi^{\ast}) < \Envy_{v'}(G, \pi^{\ast}) \leq \Envy_{v'}(G, \pi') < \Envy_v(G, \pi') + \varepsilon.
\end{equation*}

Therefore, as long as $\varepsilon$ is less than the minimum nonzero difference in envy between any two allocations for $v$, $\pi^{\ast}$ is optimal for $v$.
\end{proof}

\begin{proposition}\label{prop:knarbitrary}
When $G$ is the complete graph $K_n$, and agents have arbitrary (non-identical) valuation functions, a minimum envy allocation can be computed in polynomial time.
\end{proposition}
\begin{proof}
Given an instance of the house allocation problem $N, H, \{v_i\}_{i \in N}$, we construct a weighted bipartite matching instance $\hat G$ on the set of vertices $N \cup H$ as follows: for $i \in N$ and $h \in H$, the edge $(i, h)$ has weight $\sum_{h' \in H \setminus h} \max\{v_i(h') - v_i(h), 0\}$. 

Any perfect matching in $\hat G$ corresponds to an allocation $\pi$ with weight
\begin{equation*}
    \sum_{i \in N} \sum_{h' \in H \setminus \pi(i)} \max\{v_i(h') - v_i(\pi(i)), 0\},
\end{equation*}
which is equal to the total envy of the allocation $\pi$. Therefore, computing a minimum envy allocation is equivalent to computing a minimum weight maximum cardinality matching in $\hat G$. It is well-known that this can be done in polynomial time \citep{ahuja1988network}.
\end{proof}

\section{Proofs from Section \ref{sec:lowerbounds}}\label{apdx:lowerbounds}

\npcompletebipartition*
\begin{proof}
In the decision version of the \textsc{Minimum Bisection Problem}, we are given an instance $\langle G, k\rangle$, and we ask if there is a bisection of $G$ with $k$ or fewer edges crossing the cut. Given such an instance, we construct an instance of the graphical house allocation problem as follows. We use the same graph $G$, and our valuation interval has a cluster of $n/2$ values (within a subinterval of length $\varepsilon$) and a cluster of $n/2$ values (within a subinterval of length $\varepsilon$), where $n$ is the number of vertices of $G$. We will choose $\varepsilon$ later.

We claim that there is a bisection of $G$ with $k$ or fewer edges crossing the cut if and only if there is an allocation in our instance with total envy $k + n^2\varepsilon$. The forward direction is trivial, just by allocating values to $G$ in accordance with the bisection. To see the converse, note that if the total envy is not much more than $k$, then not many edges can cross the cut defined by the $0$-valued agents on one side and the $1$-valued agents on the other.

To make this condition true, we set $\varepsilon \approx n^{-3}$. Note that this is a linear time reduction.
\end{proof}

\thmdisjointnpcomplete*
\begin{proof}
In the \textsc{Unary Bin Packing} problem, we are given a set $I$ of items, item sizes $s(i) \in \mathbb{Z}^+$ for all $i \in I$, a bin size $B$, and a target integer $k$. The problem asks, does there exist a \emph{packing} of the items into at most $k$ bins? A packing is simply a partitioning of the items into the bins, such that for any bin, the sum of the sizes of its constituent items does not exist the bin size $B$.

Given an arbitrary instance $\langle I, s(\cdot), B, k\rangle$ of \textsc{Unary Bin Packing}, we create an instance of the graphical house allocation problem as follows. Fix some very large $C$ and some very small $\varepsilon > 0$, and let $n = kB$. For each item $i \in I$, take the graph in $\mathcal{A}$ that is the image of $s(i)$, and let $G$ be the disjoint union of all of these graphs. Note that $G \in \mathcal{G}$. Define $H = \{h_1, \ldots, h_n\}$, and for the valuation interval, define (identical) valuations $v(h_j) = \floor{\frac{j}{B}}*C + \varepsilon j$.

We wish to show that the given instance is in \textsc{Unary Bin Packing} if and only if the graphical house allocation instance (possibly padded with isolated vertices to add up to $kB$) has an allocation with envy less than $C$.

The forward direction is trivial; for the packing that attains the capacity constraints, put the graphs in the corresponding ``clusters'' on the valuation interval, putting the isolated vertices on the remaining values. This attains envy much smaller than $C$, as long as $\varepsilon$ is small enough.

Conversely, if the envy is smaller than $C$, then no edge crosses any of the ``large'' intervals in between two consecutive clusters. Therefore, each connected component can be mapped to a particular cluster on the valuation interval. Simply put the corresponding item in the corresponding bin to obtain a packing.

Note that this is a polynomial time reduction, as the bin packing instance was given in unary.
\end{proof}




\section{Proofs from Section \ref{sec:connected}}\label{apdx:connected}


\thmpath*
\begin{proof}
The result is trivial when $|H| \leq 2$, so suppose $|H| > 2$. Fix an arbitrary allocation $\pi$, and observe that $h_1$ and $h_n$ (the minimum and maximum-valued houses) have to placed on some two vertices along the path $P_n$. Suppose the sub-path between them is $(i_1, \ldots, i_k)$, with $\pi(i_1) = h_1$ and $\pi(i_k) = h_n$ without loss of generality. Then, the envy along that sub-path is, using the triangle inequality repeatedly,
\begin{align*}
    \sum_{r = 1}^{k-1}|v(\pi(i_{r+1})) - v(\pi(i_r))| &\geq |v(\pi(i_k)) - v(\pi(i_1))| \\
    &= v(h_n) - v(h_1)
\end{align*}
It follows that $\Envy(\pi, P_n) \geq v(h_n) - v(h_1)$ for all allocations $\pi$. It is straightforward to see that this minimum is attained by sorting the houses in order along the path, and furthermore, this is unique.
\end{proof}

\thmcycle*
\begin{proof}
The result is trivial when $|H| \leq 3$, so suppose $|H| > 3$. Fix an arbitrary allocation $\pi$, and observe that $h_1$ and $h_n$ have to placed on some two vertices on the cycle $C_n$. As in the proof of Theorem \ref{thm:path}, we know each of the two paths along the cycle from $h_1$ to $h_n$ have to have envy at least $v(h_n) - v(h_1)$, and so $\Envy(\pi, C_n) \geq 2(v(h_n) - v(h_1))$ for all allocations $\pi$. Once again, it is straightforward to see that this minimum is attained by sorting the houses in order along each of the two paths.
\end{proof}

\corcycle*
\begin{proof}
We fix an arbitrary agent in $C_n$ to assign $h_1$ to. Subsequently, we can choose an arbitrary subset of $H\setminus\{h_1, h_n\}$ to appear along one of the paths to $h_n$. Note that this choice completely determines an optimal allocation, as the other path contains the complement of the selected subset, and each of the subsets appears in sorted order along the paths. The number of such subsets is $2^{n-2}$. Since choosing the complement of our selected subset would have given us the same allocation up to a reversal and rotation, we have over-counted by a factor of two, and the result follows.
\end{proof}

\thmbipartiteequal*
\begin{proof}
For clarity let the two parts of the bipartite graph be $L$ and $R$ respectively ($|L| = |R| = r$). We refer to the property in the theorem statement as the {\em optimal} property. 

Let $n^{>}_{L, \pi}(x)$ be the number of agents in $L$ who receive a value strictly greater than $x$ in $\pi$. We similarly define $n^{<}_{L, \pi}(x), n^{>}_{R, \pi}(x)$ and $n^{<}_{R, \pi}(x)$.

Assume for contradiction that some optimal allocation $\pi^{\ast}$ does not satisfy the optimal property. This means there exists an $i \in [m]$ such that both $h_{2i-1}$ and $h_{2i}$ are allocated to the same part. Let $j$ be the least such $i$ where this is true and assume, without loss of generality,  $h_{2j-1}$ and $h_{2j}$ are allocated to agents in $L$.

Let $\{h_{2j-1}, h_{2j}, \dots, h_{2j+k}\}$ be the set of houses allocated to agents in $L$ such that $h_{2j+k+1}$ is allocated some agent in $R$. Note that by our assumption, we have $k \ge 0$.

Construct an allocation $\pi'$ starting at $\pi^{\ast}$ and swapping the house $h_{2j+k}$ and $h_{2j+k+1}$. Essentially, we are swapping a house allocated to some agent in $L$ with a house allocated to some agent in $R$.  

Let us compute the difference in envy between allocations $\pi^{\ast}$ and $\pi'$. For any agent with value lesser than $v(h_{2j+k})$ under $\pi^{\ast}$ in $L$, their envy towards their neighbors in $\pi'$ is less than their envy in $\pi^{\ast}$ by exactly $v(h_{2j+k+1}) - v(h_{2j+k})$. Similarly, for any agent with value greater than $v(h_{2j+k})$ under $\pi^{\ast}$ in $L$, their envy towards their neighbors in $\pi'$ is greater than their envy in $\pi^{\ast}$ by exactly $v(h_{2j+k+1}) - v(h_{2j+k})$. Extending this reasoning, we get the following expression for difference in envy.
\begin{align*}
&\Envy(\pi', G) - \Envy(\pi^{\ast}, G) \\
&= [n^{>}_{L, \pi^{\ast}}(v(h_{2j+k})) - n^{<}_{L, \pi^{\ast}}(v(h_{2j+k}))](v(h_{2j+k+1}) - v(h_{2j+k})) \\
&\qquad \qquad + [n^{<}_{R, \pi^{\ast}}(v(h_{2j+k+1})) - n^{>}_{R, \pi^{\ast}}(v(h_{2j+k+1}))](v(h_{2j+k+1}) - v(h_{2j+k})) \\
&= [n^{>}_{L, \pi^{\ast}}(v(h_{2j+k})) - n^{<}_{L, \pi^{\ast}}(v(h_{2j+k})) + n^{<}_{R, \pi^{\ast}}(v(h_{2j+k+1})) \\
&\qquad \qquad - n^{>}_{R, \pi^{\ast}}(v(h_{2j+k+1}))](v(h_{2j+k+1}) - v(h_{2j+k})) \\
&= [(r - (k+2 + j-1)) - (k+1+j-1) + (j-1) - (r - j)](v(h_{2j+k+1}) - v(h_{2j+k})) \\
&=[2j - 2(k+j) - 2](v(h_{2j+k+1}) - v(h_{2j+k})) \\
&= [-2k - 2](v(h_{2j+k+1}) - v(h_{2j+k})) \\
&<0
\end{align*}
The third equality follows from our choice of $j$; for any $i < j$, exactly one of $h_{2i-1}$ and $h_{2i}$ is allocated to $L$ under $\pi^{\ast}$.
The inequality follows since $k \ge 0$ and $v(h_{2j+k+1}) - v(h_{2j+k}) > 0$.
This implies that $\pi'$ has a lower envy than $\pi^{\ast}$ contradicting the optimality of $\pi^{\ast}$.
We can therefore assume that all minimum envy allocations have the optimal property. 

We prove the second part of the claim by showing that for any allocation that satisfies the optimal property, for any $i \in [m]$, swapping $h_{2i-1}$ and $h_{2i}$ results in an allocation with equal envy. This observation can be repeatedly applied to show that any two allocations that satisfy the optimal property have the same envy. Note that permuting the allocation to a specific part ($L$ or $R$) does not effect the total envy in any way as well. 

More formally, let $\pi$ be any allocation that satisfies the optimal property. Pick an arbitrary $i \in [m]$ and swap $h_{2i-1}$ and $h_{2i}$ to create allocation $\pi'$; we assume without loss of generality $h_{2i-1}$ is allocated to some agent in $L$ in $\pi$. The difference in envy of the two allocations is given by:
\begin{align*}
&\Envy(\pi', G) - \Envy(\pi, G) \\
&= [n^{>}_{L, \pi}(v(h_{2i-1})) - n^{<}_{L, \pi}(v(h_{2i-1}))](v(h_{2i}) - v(h_{2i-1})) \\
&\qquad \qquad + [n^{<}_{R, \pi}(v(h_{2i})) - n^{>}_{R, \pi}(v(h_{2i}))](v(h_{2i}) - v(h_{2i-1})) \\
&= [n^{>}_{L, \pi}(v(h_{2i-1})) - n^{<}_{L, \pi}(v(h_{2i-1})) + n^{<}_{R, \pi}(v(h_{2i})) - n^{>}_{R, \pi}(v(h_{2i}))](v(h_{2i}) - v(h_{2i-1})) \\
&= [(r-i) - (i-1) + (i-1) - (r - i)](v(h_{2i}) - v(h_{2i-1})) \\
&=0
\end{align*}
\end{proof}

\thmbipartiteunequal*
\begin{proof}
This proof is very similar to that of Theorem \ref{thm:bipartite-equal}. Let the two parts of the bipartite graph be $L$ and $R$ respectively such that $r = |L| > |R| = s$. We refer to the properties in the theorem statement when $r - s$ is even and odd as the {\em optimal even} property and the {\em optimal odd} property respectively. 

Let $n^{>}_{L, \pi}(x)$ be the number of agents in $L$ who receive a value strictly greater than $x$ in $\pi$. We similarly define $n^{<}_{L, \pi}(x), n^{>}_{R, \pi}(x)$ and $n^{<}_{R, \pi}(x)$.

{\textbf{Case 1:} $r - s$ is even.}
We split the proof into two claims
\begin{claim}\label{claim:case-1-part-1}
Any optimal allocation allocates the first $m$ houses to agents in $L$.
\end{claim}
\begin{proof}
Assume for contradiction that this is not true. That is, there is an optimal allocation $\pi$ such that:
\begin{align*}
    &\pi(h_j) \in L \text{ for all } j \in [k] \text{ for some } \frac{r - s}{2} > k \ge 0 \\
    &\pi(h_{k+j}) \in R \text{ for all } j \in [l] \text{ for some } l > 0 \\
    &\pi(h_{k+l+1}) \in L
\end{align*}
Let us swap $h_{k+l}$ and $h_{k+l+1}$ to obtain an allocation $\pi'$. We can now compare the aggregate envy of $\pi$ and $\pi'$ using arguments similar to that of Theorem \ref{thm:bipartite-equal}.

\begin{align*}
&\Envy(\pi', G) - \Envy(\pi, G) \\
&= [n^{<}_{L, \pi}(v(h_{k+l+1})) - n^{>}_{L, \pi}(v(h_{k+l+1}))](v(h_{k+l+1}) - v(h_{k+l})) \\
&\qquad \qquad + [n^{>}_{R, \pi}(v(h_{k+l})) - n^{<}_{R, \pi}(v(h_{k+l}))](v(h_{k+l+1}) - v(h_{k+l})) \\
&= [n^{<}_{L, \pi}(v(h_{k+l+1})) - n^{>}_{L, \pi}(v(h_{k+l+1})) + n^{>}_{R, \pi}(v(h_{k+l})) - n^{<}_{R, \pi}(v(h_{k+l}))](v(h_{k+l+1}) - v(h_{k+l})) \\
&= [k - (r - (k+1)) + (s - l) - (l-1)](v(h_{2i}) - v(h_{2i-1})) \\
&= 2k - (r -s) + 2 - 2l \\
&<0
\end{align*}
The last inequality follows from the fact that $l \ge 1$ and $k < (r - s)/2$. This contradicts the optimality of $\pi$.
\end{proof}
\begin{claim}\label{claim:case-1-part-2}
In any optimal allocation, for any $i \in [s]$, $h_{m+ 2i-1}$ and $h_{m+ 2i}$ cannot be allocated to the same part. 
\end{claim}
\begin{proof}
Assume for contradiction that this is not true. Let $\pi$ be an optimal allocation that satisfies Claim \ref{claim:case-1-part-1} but not Claim \ref{claim:case-1-part-2}. Choose $j$ as the least $i$ such that $h_{m+ 2i-1}$ and $h_{m + 2i}$ are allocated to the same part. Assume they are both allocated to some agents in $L$. The proof for $R$ can be shown similarly. Let $\{h_{m+ 2j-1}, h_{m + 2j}, \dots, h_{m + 2j+k}\}$ be a set of houses allocated to agents in $L$ such that $h_{m + 2j+k+1}$ is allocated to some agent in $R$. Let $\pi'$ be the allocation that results from swapping $h_{m + 2j+k}$ and $h_{m + 2j+k+1}$ in $\pi$. We can compare the envy between $\pi'$ and $\pi$.
\begin{align*}
&\Envy(\pi', G) - \Envy(\pi, G) \\
&= [n^{>}_{L, \pi}(v(h_{m+2j+k})) - n^{<}_{L, \pi}(v(h_{m+2j+k}))](v(h_{m+2j+k+1}) - v(h_{m+2j+k})) \\
&\qquad \qquad + [n^{<}_{R, \pi}(v(h_{m+2j+k+1})) - n^{>}_{R, \pi}(v(h_{m+2j+k+1}))](v(h_{m+2j+k+1}) - v(h_{m+2j+k})) \\
&= [n^{>}_{L, \pi}(v(h_{m+2j+k})) - n^{<}_{L, \pi}(v(h_{m+2j+k})) + n^{<}_{R, \pi}(v(h_{m+2j+k+1})) \\
&\qquad \qquad - n^{>}_{R, \pi}(v(h_{m+2j+k+1}))] (v(h_{m+2j+k+1}) - v(h_{m+2j+k})) \\
&= [(r - ((r-s)/2 + k+2 + j-1)) - ((r-s)/2 + k+1+j-1) \\
&\qquad \qquad + (j-1) - (s - j)] (v(h_{m+2j+k+1}) - v(h_{m+2j+k})) \\
&=[2j - 2(k+j) - 2](v(h_{m+2j+k+1}) - v(h_{m+2j+k})) \\
&= [-2k - 2](v(h_{m+2j+k+1}) - v(h_{m+2j+k})) \\
&<0
\end{align*}
The final inequality holds since $k \ge 0$. Again, we contradict the optimality of $\pi$.
\end{proof}
Claim \ref{claim:case-1-part-2} also implies that none of the final $(r - s)/2$ houses are allocated to agents in $R$. We can therefore conclude that these houses must be allocated to agents in $L$ in any optimal allocation. 

To show that any allocation that satisfies the optimal even property has the same aggregate envy, we use a swapping based argument similar to Theorem \ref{thm:bipartite-equal}. Let $\pi$ be any allocation that satisfies the optimal even property. Pick an arbitrary $i \in [s]$ and let $\pi'$ be the allocation that results from swapping $h_{m + 2i - 1}$ and $h_{m + 2i}$ in $\pi$. Assume that $h_{m + 2i - 1}$ is allocated to $L$ in $\pi$. The proof for $R$ flows similarly. Let us compare the envy of the two allocations. 

\begin{align*}
&\Envy(\pi', G) - \Envy(\pi, G) \\
&= [n^{>}_{L, \pi}(v(h_{m+2i-1})) - n^{<}_{L, \pi}(v(h_{m + 2i-1}))](v(h_{m + 2i}) - v(h_{m + 2i-1})) \\
&\qquad \qquad + [n^{<}_{R, \pi}(v(h_{m + 2i})) - n^{>}_{R, \pi}(v(h_{m + 2i}))](v(h_{m + 2i}) - v(h_{m + 2i-1})) \\
&= [n^{>}_{L, \pi}(v(h_{m + 2i-1})) - n^{<}_{L, \pi}(v(h_{m + 2i-1})) + n^{<}_{R, \pi}(v(h_{m + 2i})) \\
&\qquad \qquad - n^{>}_{R, \pi}(v(h_{m + 2i}))](v(h_{m + 2i}) - v(h_{m + 2i-1})) \\
&= [(r- (i+m)) - (m + i-1) + (i-1) - (s - i)](v(h_{m + 2i}) - v(h_{m + 2i-1})) \\
&=0
\end{align*}

{\textbf{Case 2:} $r - s$ is odd}

This proof is, surprisingly, very similar to the previous case. We similarly split the proof into two claims. 

\begin{claim}\label{claim:case-2-part-1}
Any optimal allocation allocates the first $m$ houses to agents in $L$.
\end{claim}
The proof to this claim is exactly the same as the proof to the Claim \ref{claim:case-1-part-1}.
So we move on to the second claim. 

\begin{claim}\label{claim:case-2-part-2}
In any optimal allocation, for any $i \in [s]$, $h_{m+ 2i-1}$ is allocated to some agent in $L$ and $h_{m+ 2i}$ is allocated to some agent in $R$. 
\end{claim}
\begin{proof}
This proof is again very similar to Claim \ref{claim:case-1-part-2}. However, there are some subtle differences.

Assume for contradiction that the claim is not true. Let $\pi$ be an optimal allocation that satisfies Claim \ref{claim:case-2-part-1} but not Claim \ref{claim:case-2-part-2}. Choose $j$ as the least $i$ where the claim is violated. That is, either $h_{m + 2j-1}$ is allocated to $R$ or $h_{m + 2j}$ is allocated to $L$. In this proof, we assume the latter has occured. The proof for the former is very similar. In other words, both $h_{m+ 2j-1}$ and $h_{m+ 2i}$ are allocated to some agents in $L$. Let $h_{m+ 2j-1}, h_{m + 2j}, \dots, h_{m + 2j+k}$ be a set of houses allocated to agents in $L$ such that $h_{m + 2j+k+1}$ is allocated to some agent in $R$. Let $\pi'$ be the allocation that results from swapping $h_{m + 2j+k}$ and $h_{m + 2j+k+1}$. We can compare the envy between $\pi'$ and $\pi$.
\begin{align*}
&\Envy(\pi', G) - \Envy(\pi, G) \\
&= [n^{>}_{L, \pi}(v(h_{m +2j+k})) - n^{<}_{L, \pi}(v(h_{m+2j+k}))](v(h_{m+2j+k+1}) - v(h_{m+2j+k})) \\
&\qquad \qquad + [n^{<}_{R, \pi}(v(h_{m+2j+k+1})) - n^{>}_{R, \pi}(v(h_{m+2j+k+1}))](v(h_{m+2j+k+1}) - v(h_{m+2j+k})) \\
&= [n^{>}_{L, \pi}(v(h_{m+2j+k})) - n^{<}_{L, \pi}(v(h_{m+2j+k})) + n^{<}_{R, \pi}(v(h_{m+2j+k+1})) \\
&\qquad \qquad - n^{>}_{R, \pi}(v(h_{m+2j+k+1}))](v(h_{m+2j+k+1}) - v(h_{m+2j+k})) \\
&= [(r - (m + k+2 + j-1)) - (m + k+1+j-1) \\
&\qquad \qquad + (j-1) - (s - j)] (v(h_{m+2j+k+1}) - v(h_{m+2j+k})) \\
&=[2j - 2(k+j) - 1](v(h_{m+2j+k+1}) - v(h_{m+2j+k})) \\
&= [-2k - 1](v(h_{m+2j+k+1}) - v(h_{m+2j+k})) \\
&<0
\end{align*}
The final inequality holds since $k \ge 0$. The optimality of $\pi$ has been contradicted. 
\end{proof}

Claim \ref{claim:case-2-part-2} also implies that none of the final $m+1$ houses are allocated to agents in $R$. We can therefore conclude that these houses must be allocated to agents in $L$ in any optimal allocation. 

Note that the optimal odd property specifies exactly which houses must be allocated to $L$ and $R$ in any optimal allocation. Any two allocations which satisfy the optimal odd property can only differ over which agents in $L$ and $R$ houses are allocated to and not which houses are allocated to $L$ and $R$. It is easy to see that this difference cannot lead to a difference in envy in the case of the complete bipartite graph.
\end{proof}

\corbipartite*
\begin{proof}
For clarity, we refer to the larger part of the bipartite graph as $L$ and the smaller part as $R$.

In this proof, we are trying to count the number of allocations which allocate a different set of houses to $L$ (and therefore, $R$ as well). There are of course, $r!$ allocations given a set of houses to allocate to agents in $L$ but we ignore this factor. 

When $r - s$ is even, there are $s$ different choices we can make. That is, for each $i \in [s]$, we can choose which of $h_{m + 2i - 1}$ and $h_{m + 2i}$ goes to $L$ and which one goes to $R$ (Theorem \ref{thm:bipartite-unequal}). This gives us $2^{s}$ different allocations.

When $r - s$ is odd, there is no choice since Theorem \ref{thm:bipartite-unequal} shows that only one specific set of houses allocated to $L$ achieves the optimal envy. Therefore, not counting permutations over allocations to the same part, there is only one unique allocation. 
\end{proof}

\lemvaluationinverse*
\begin{proof}
For any edge $(i, j)$ in the graph $G$ and and any allocation $\pi$, we have
\begin{align*}
    |v(\pi(i)) - v(\pi(j))| = |(-v(\pi(i))) - (-v(\pi(j)))| = |v^{inv}(\pi(i)) - v^{inv}(\pi(j))| 
\end{align*}
\end{proof}

\lemlocalmedianimprovement*
\begin{proof}
\begin{figure}[ht]
    \begin{subfigure}[b]{0.45\textwidth}
        \centering
        \begin{tikzpicture}[,mycirc/.style={circle,fill=white, draw = black,minimum size=0.75cm,inner sep = 3pt}]
        \node[mycirc] {$r$} 
            child {node[mycirc] {$y$}
                child[solid] {node[mycirc] {$x_m$}  
                    child {node[mycirc] {$x_{m-1}$}
                        child {node[mycirc] {$\ldots$}
                            child {node[mycirc] {$x_{1}$}
                                child {node[mycirc] {$u$} edge from parent [dashed]}
                                child {node[mycirc] {$v$} edge from parent [dashed]}}
                            child {node[mycirc] {$\ldots$}}}
                        child {node[mycirc] {$z_{m-2}$}}}
                    child {node[mycirc] {$z_{m-1}$}}}
                child[solid] {node[mycirc] {$z_m$}} edge from parent [dashed]};
        \end{tikzpicture}
        \caption{Before the swap (allocation $\pi$)}
        \label{fig:node-before-swap}
    \end{subfigure}
    \begin{subfigure}[b]{0.45\textwidth}
        \centering
        \begin{tikzpicture}[,mycirc/.style={circle,fill=white, draw = black,minimum size=0.75cm,inner sep = 3pt}]
        \node[mycirc] {$r$} 
            child {node[mycirc] {$x_m$}
                child[solid] {node[mycirc] {$x_{m-1}$}  
                    child {node[mycirc] {$x_{m-2}$}
                        child {node[mycirc] {$\ldots$}
                            child {node[mycirc] {$y$}
                                child {node[mycirc] {$u$} edge from parent [dashed]}
                                child {node[mycirc] {$v$} edge from parent [dashed]}}
                            child {node[mycirc] {$\ldots$}}}
                        child {node[mycirc] {$z_{m-2}$}}}
                    child {node[mycirc] {$z_{m-1}$}}}
                child[solid] {node[mycirc] {$z_m$}} edge from parent [dashed]};
        \end{tikzpicture}
        \caption{After the swap (allocation $\pi'$).}
        \label{fig:node-after-swap}
    \end{subfigure}
    \caption{Cyclic swap to show the local median property holds (Lemma \ref{lem:local-median-improvement}). Solid edges are guaranteed to exist. Dashed edges may or may not exist.
    }
    \label{fig:local-median-swap}
\end{figure}
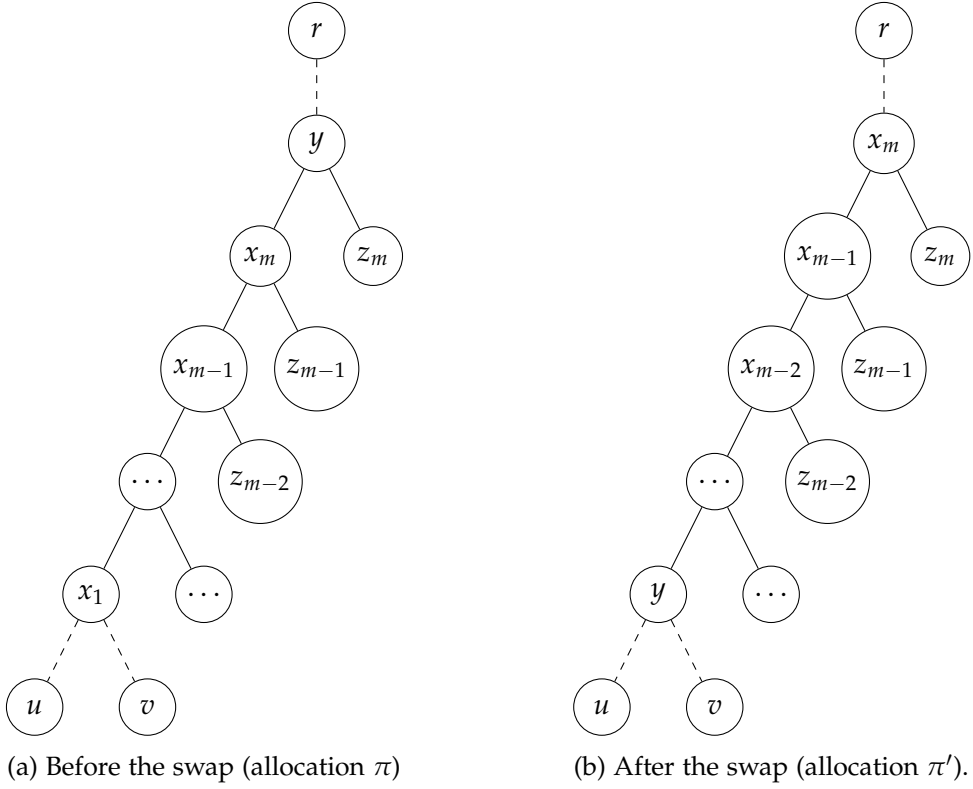

Let the node $i$ have value $y$ under $\pi$ and its children have values $x_m$ and $z_m$ respectively under $\pi$. By assumption, either $y < \min\{x_m, z_m\}$ or $y > \max\{x_m, z_m\}$. We show that in either case, the lemma holds. Since allocations are bijective and the values can be assumed to be distinct, we will refer to tree nodes using the value allocated to them in $\pi$.

{\textbf{Case 1: }($y < \min\{x_m, z_m\}$).} Assume WLOG that $x_m < z_m$. We construct a path recursively as follows. Initialize the path as $(y)$. If the final node on the path either has no children or has at least one child with allocated value lower than $y$, then stop. Otherwise, pick the least valued child of the final node on the path and append it to the path. This gives us a path $(y, x_m, x_{m-1}, \dots, x_1)$ for some nodes with value $x_m, x_{m-1}, \dots, x_1$ in $T$. Note that by definition, this path has at least $2$ vertices, i.e., $m \ge 1$. We construct a new allocation $\pi'$ from $\pi$ by cyclically transferring houses as follows: we give the agent with value $y$ the house with value $x_m$, we give the agent with value $x_m$ the house with value $x_{m-1}$ and so on till finally we give the agent with value $x_1$ the house with value $y$. This has been described in Figure \ref{fig:local-median-swap}.

The solid edges and dashed edges in Figure \ref{fig:local-median-swap} cover all possible edges $e$ in $T$ where $\envy_{\pi}(e) \ne \envy_{\pi'}(e)$.
From our path construction and our assumption that $i$ is the furthest node from the root which does not satisfy the local median property, we have the following two properties: 
\begin{inparaenum}[(a)]
    \item $\max\{x_j, z_j\} > x_{j+1} > \min\{x_{j}, z_j\}$ for all $j \in [m-1]$, and 
    \item $x_j < z_j$ for all $j \in [m]$.
\end{inparaenum}
These two properties allow us to compare the envy along the solid edges:
\begin{align*}
    \envy_{\pi}({\rm solid}) &= (x_m - y) + (z_m - y) + \left [\sum_{j = 2}^{m} ((x_{j} - x_{j-1}) + (z_{j-1} - x_j)) \right ] \\
    &= \left [\sum_{j = 1}^{m-1} (z_j - x_{j}) \right ] + (z_m - y) + (x_m - y). \\
    \envy_{\pi'}({\rm solid}) &= (x_1 - y) + (z_1 - x_1) + \left [\sum_{j = 2}^m ((x_{j} - x_{j-1}) + (z_j - x_j)) \right ] \\
    &= \left [\sum_{j = 1}^{m} (z_j - x_{j}) \right ] + (x_m-x_1) + (x_1 - y).
\end{align*}

Combining the two values, we get
\begin{align*}
    \envy_{\pi'}({\rm solid}) - \envy_{\pi}({\rm solid}) = y - x_m.
\end{align*}

To compute the difference in envy along the dashed lines, some straightforward casework is required. There are many different possible relations between $u$, $v$, $x_1$, and $y$, and between $r$, $y$, and $x_m$.
All possible cases and their corresponding results are summarized in Table \ref{tab:local-median-case-work}. 
There are two assumptions made in Table \ref{tab:local-median-case-work}. First, WLOG we assume $u < v$. Second, $u < y$, since this is the termination condition from our path construction. 

\begin{table}[ht]
    \centering
    \begin{tabular}{|c|c|c|c|c|}
        \hline
        Cases & $r$ does not exist & $r < y < x_m$ & $y < r < x_m$ & $y < x_m < r$ \\
        \hline
        $u$ and $v$ do not exist & $0$ & $(x_m - y)$ & $< (x_m - y)$ & $(y - x_m)$ \\
        $u < y < v < x_1$ & $< 0$ & $ < (x_m - y)$ & $< (x_m - y)$ & $ < (y - x_m)$ \\
        $u < y < x_1 < v$ & $0$ & $(x_m - y)$ & $< (x_m - y)$ & $(y - x_m)$ \\
        $u < v < y < x_1$ & $< 0$ & $ < (x_m - y)$ & $< (x_m - y)$ & $ < (y - x_m)$ \\
        \hline
    \end{tabular}
    \caption{Cases for the possible values of $\envy_{\pi'}({\rm dashed}) - \envy_{\pi}({\rm dashed})$}
    \label{tab:local-median-case-work}
\end{table}

If $y$ is the root of the tree and $r$ does not exist, from Table \ref{tab:local-median-case-work}, we get: 
\begin{align*}
    \Envy&(\pi', T) - \Envy(\pi, T) \\
    &= \envy_{\pi'}({\rm solid}) - \envy_{\pi}({\rm solid}) + \envy_{\pi'}({\rm dashed}) - \envy_{\pi}({\rm dashed}) \\
    &\le (y - x_m) + 0 \\
    &< 0.
\end{align*}
Therefore, the total envy of $\pi'$ is strictly less than that of $\pi$. Note that this implies that, even if $y$ has a parent, the total envy along the subtree rooted at $y$ (denoted by $T'$) strictly reduces.

Coming to the case where $r$ exists, let us study the envy of any tree $T''$ that is not contained by $T'$. Either $T''$ contains $T'$, or $T''$ and $T'$ are disjoint. If they are disjoint, then $\Envy(\pi, T'') = \Envy(\pi', T'')$, since the allocation on the subtree $T''$ is the same in $\pi$ and $\pi'$. If $T''$ strictly contains $T'$, $T''$ must contain the node $r$. From Table \ref{tab:local-median-case-work}, we get:
\begin{align*}
    \Envy&(\pi', T'') - \Envy(\pi, T'') \\
    &= \envy_{\pi'}({\rm solid}) - \envy_{\pi}({\rm solid}) + \envy_{\pi'}({\rm dashed}) - \envy_{\pi}({\rm dashed}) \\
    &\le (y - x_m) + (x_m - y) \\
    &= 0.
\end{align*}
Therefore the total envy weakly decreases and we are done.

{\textbf{Case 2: }($y > \max\{x_m, z_m\}$).} 
This implies $- y < \min\{-x_m, -z_m\}$.
We can therefore apply Case 1 to the allocation $\pi$  under the inverted valuations $v^{\rm inv}$. 
It follows that, with respect to $v^{\rm inv}$, there is a an allocation $\pi'$ which has a strictly lower total envy along the subtree $T'$ rooted at $i$ and a weakly lower total envy along any subtree $T''$ that is not contained by $T'$. 
Applying Lemma \ref{lem:valuation-inverse} with the allocations $\pi'$ and $\pi$, we get the required result.
\end{proof}

\section{Proofs from Section \ref{sec:disconnected}}\label{apdx:disconnected}

\propinseparable*
\begin{proof}
The graph given in Figure \ref{fig:degreeseparated} is an inseparable forest. 

Assume that the clusters along the valuation interval are sufficiently packed (each within a subinterval of length $\varepsilon$), and furthermore, they are equispaced along the entire valuation interval, and WLOG assume the entire valuation interval has length $1$. 
Note that the allocation that places the induced stars of the given graph in the corresponding clusters along the valuation interval attains a total envy of at most $4/3 + 0.001$ (assuming $\varepsilon$ is small enough). Let the four vertices of degree $2$ or more be $x_1, x_2, x_3, x_4$, where $x_i$ is incident to exactly $s_i$ degree-$1$ vertices. Let us also number the clusters along the valuation interval $1, 2, 3, 4$ from left to right.

We first claim that in any optimal allocation, $x_i$ cannot be in cluster $j$ for $j < i$. Otherwise, at least three of the $s_i$ neighbors of $x_i$ must lie in other clusters, so one of the three large subintervals must be counted three or more times. It is easy to see that then the envy exceeds $5/3$. We next claim that $x_i$ and $x_j$ cannot be in the same cluster, for $i \neq j$. Otherwise, again, at least three edges pass over the same large subinterval, and so the envy exceeds $5/3$ again.

It follows that $x_i$ must belong to the $i$th cluster, for all $i$. The result follows immediately.
\end{proof}

\propseparablenotstrong*
\begin{proof}
The graph $G$ given in Figure \ref{fig:spannerpoint} is a separable forest that is not strongly separable.

It is trivial that $G$ is separable, as one component is a single vertex that cannot be split on the valuation interval.

Consider the lower valuation interval that is depicted in Figure \ref{fig:spannerpoint}. Assume that the clusters along the valuation interval are sufficiently packed (each within a subinterval of length $\varepsilon$), and furthermore, the sole valuation in the middle is exactly at the center of the interval. WLOG assume the entire valuation interval has length $1$. Note that the allocation that places the induced stars of $G$ in the clusters attains a total envy of at most $1 + 0.001$ (assuming $\varepsilon$ is small enough).

We first claim that an optimal allocation cannot place both the degree-$3$ vertices in the same cluster. In such an allocation, one of the two large subintervals needs to be covered by at least two edges, and so the total envy is at least $3/2$.

We next claim that an optimal allocation cannot place a degree-$3$ vertex in the center. If it does, then again by a similar casework as in the previous paragraph, one large subinterval has to be covered by at least two edges, and so the total envy is at least $3/2$.

Our result follows immediately.
\end{proof}

\cordegreeonealgorithm*
\begin{proof}
If there are $k'$ paths of length $1$ and $m'$ isolated vertices, then suppose $\varphi(k', m', \ell)$ denotes the optimal envy using $k'$ edges and $m'$ isolated vertices on the house set $\{h_1, \ldots, h_\ell\}$. Using Theorem \ref{thm:disjointpaths}, we have the recursion
\begin{equation*}
    \varphi(k', m', \ell) = \min(\varphi(k' - 1, m', \ell - 2) + v(h_{\ell}) - v(h_{\ell - 1}), \varphi(k', m' - 1, \ell - 1)).
\end{equation*}
Dynamically solving this yields an $O(n^3)$ algorithm to find $\varphi(k, 2n - k, n)$.
\end{proof}

\thmdisjointpaths*
\begin{proof}
By Theorem \ref{thm:path}, we know that each of the paths should have its allocated houses in sorted order. Now, suppose there are values $h_k \prec h_\ell \prec h_m$, with $h_k$ and $h_m$ being allocated to $P_{n_i}$, and $h_\ell$ to a different path $P_{n_j}$. 
%
            %
            %
%
%
We can reallocate the houses only on these two paths and strictly improve the allocation. For instance, suppose $H_i := \pi(P_{n_i})$ and $H_j := \pi(P_{n_j})$. We can now allocate the $n_i$ lowest-valued houses in $H_i \cup H_j$ to $P_{n_i}$ and the $n_j$ highest-valued houses in $H_i \cup H_j$ to $P_{n_j}$, keeping the rest of the allocation the same. This leads to an allocation with strictly lower envy than before, and this concludes the proof.
\end{proof}

\cordisjointpathsalgorithm*
\begin{proof}

The result for $r = 1$ is trivial, as by Theorem \ref{thm:disjointpaths}, there is only a unique allocation. For $r > 1$, we can test all $r!$ orders of the component paths $P_{n_1}, \dots P_{n_r}$. For each order $G_1, \dots G_r$, we assign the first $|G_1|$ houses to $G_1$, the next $|G_2|$ houses to $G_2$, etc.
\end{proof}

\begin{proposition}
\label{prop:distinctpathlengths}
Let $G$ be a disjoint union of paths. If $t$ is the number of different path lengths in $G$, then we can find an optimal allocation on $G$ for any instance in time $O(n^{t+1})$.
\end{proposition}
\begin{proof}
The result for $t = 1$ is trivial. For $t > 1$, if the distinct path lengths are $n_1, \ldots, n_t$, then suppose $\varphi(r_1, \ldots, r_t, \ell)$ denotes the optimal envy using $r_i$ paths of length $n_i$, for $i = 1, \ldots, t$, on the house set $\{h_1, \ldots, h_\ell\}$. Using Theorem \ref{thm:disjointpaths}, we have the recursion
\begin{align*}
     \varphi(r_1, \ldots, r_t, \ell) &= \min\{\varphi(r_1 - 1, r_2, \ldots, r_t, \ell - n_1) + (v(h_\ell) - v(h_{\ell - n_1 + 1})), \\ & \qquad \ldots, \varphi(r_1, \ldots, r_{t_1}, r_t - 1, \ell - n_t) + (v(h_\ell) - v(h_{\ell - n_t + 1}))\}.
\end{align*}
Dynamically solving this yields an $O(n^{r + 1})$ algorithm to find the optimal allocation on the given instance.
\end{proof}

\thmdisjointstars*
\begin{proof}
We ``separate'' any two stars while improving on our objective. Consider two stars $K_{1, n_1}$ and $K_{1, n_2}$. Let $\pi$ be any optimal allocation that allocates the values $a_1, \dots, a_{n_1 + 1}$ to $K_{1, n_1}$ and $b_1, \dots, b_{n_2 + 1}$ to $K_{1, n_2}$. 

We provide a simple two step procedure that creates a new allocation $\pi'$ that allocates contiguous intervals to both stars and attains total envy at most that of $\pi$. In the first step, we simply re-arrange the values allocated to each star to ensure they satisfy the characterization for an optimal envy allocation from Theorem \ref{thm:star}. In the second step, assuming WLOG the center of $K_{1, n_1}$ has a lower value than that of $K_{1, n_2}$, we re-arrange the values allocated to the spokes of both stars by allocating the least $n_1$ values to $K_{1, n_1}$ and the greatest $n_2$ values to $K_{1, n_2}$; crucially, we do not change the value allocated to the center of either star. It is easy to see that neither of these steps can increase the total envy.

It is also easy to see that, if the stars are not allocated contiguous intervals, the above two step procedure changes the allocation and strictly reduces the envy. This shows that not allocating contiguous intervals to each star is sub-optimal.
\end{proof}

\thmdisjointequicliques*
\begin{proof}
We prove the result for the case of two cliques $K_{n/2} + K_{n/2}$. The result for $r$ cliques follows by showing that each pair of cliques must be separated from each other.

Let $(V, E)$ and $(V', E')$ be the set of vertices and edges of each copy of $K_{n/2}$. Let $\tau: V \mapsto V'$ be any bijective mapping from $V$ to $V'$.

Let $\pi$ be any allocation on $K_{n/2}+K_{n/2}$, we show that if $\pi$ does not allocate contiguous intervals to each component, we can create a better allocation $\pi'$. 

Let $a_1 < a_2 < \dots a_{n/2}$ be the values allocated to the nodes in $V$ and $b_1 < b_2 < \dots b_{n/2}$ be the values allocated to the nodes in $V'$ in some optimal allocation $\pi$. We rearrange the goods allocated to $V'$ such that if node $v \in V$ receives $a_i$, then node $\tau(v)$ receives $b_{n/2 - i}$. This does not change the total envy of the allocation.

If each component is not allocated a contiguous interval, the least-valued $n/2$ houses must have some $a$ values and some $b$ values. Let's call the least-valued $n/2$ houses $H'$ and let's say there are $k$ $a_i$'s in $H'$. Therefore $H'$ contains $a_1, a_2, \dots, a_k$ and $b_1, b_2, \dots, b_{n/2-k}$. 

We create a new allocation $\pi'$ from $\pi$ as follows. For all $i \in [k]$, we swap $a_i$ with $b_{n/2-i}$. Note that for each house among the least-valued $n/2$ houses, if $a_i$ is allocated to $v \in V$, we swap the houses given to $v$ and $\tau(v)$.

Let us now compute the change in envy between $\pi'$ and $\pi$. We do this by showing that, for every edge $(u, v) \in E$, the total sum of the envies along the edges $(u, v)$ and $(\tau(u), \tau(v))$ decreases. 

\textbf{Case 1: $u$ and $v$ are unaffected by the swap.} Then $\tau(u)$ and $\tau(v)$ are unaffected as well. Therefore the total envy along these two edges does not change.

\textbf{Case 2: $u$ and $v$ are both affected by the swap.} Then, $\envy_{\pi'}(u,v) = \envy_{\pi}(\tau(u),\tau(v))$ and $\envy_{\pi}(u,v) = \envy_{\pi'}(\tau(u),\tau(v))$. Therefore, the total envy along these two edges does not change.

\textbf{Case 3: Only $u$ is affected by the swap.} This means $\tau(v)$ is not affected by the swap. The total envy along these two edges under $\pi$ is 

\begin{align*}
  \envy_{\pi}(u,v) + \envy_{\pi}(\tau(u),\tau(v)) = (a_j - a_i) + (b_{n/2-i} - b_{n/2-j})  
\end{align*}
where $j > k \ge i$. This can be re-written as
\begin{align*}
  \envy_{\pi}(u,v) + \envy_{\pi}(\tau(u),\tau(v)) &= 
  2\min\{a_j, b_{n/2 -i}\} + |a_j - b_{n/2-i}| \\
  &\qquad\qquad - 2\min\{a_i, b_{n/2 -j}\} - |a_i - b_{n/2-j}|
\end{align*}

The total envy along these two edges under $\pi'$ is
\begin{align*}
  \envy_{\pi'}(u,v) + \envy_{\pi'}(\tau(u),\tau(v)) = |a_j - b_{n/2-i}| + |a_i - b_{n/2-j}|  
\end{align*}
The change in envy is 
\begin{align*}
   2\min\{a_i, b_{n/2 -j}\} -  2\min\{a_j, b_{n/2 -i}\} < 0
\end{align*}
The inequality holds since $j > k \ge i$.

When $k \ge 1$, at least one edge belongs to Case 3 and so the total envy of $\pi'$ is strictly less than the total envy of $\pi$.
\end{proof}

\cordisjointequicliquesalgorithm*
\begin{proof}
As demonstrated in Theorem \ref{thm:disjointequicliques}, we just need to assign the first $\frac{n}{r}$ houses to clique $1$, the next $\frac{n}{r}$ houses to clique $2$, and so on. This can be done in linear time.
\end{proof}

\thmdisjointcliques*
\begin{proof}
Let $\pi$ be any minimum envy allocation. Assume for contradiction that there exist two cliques ($K$ and $K'$ w.l.o.g.) such that $|K| > |K'|$ and $K$ does not receive a contiguous set of valuations with respect to the houses in $K \cup K'$. The case where $|K| = |K'|$ has been shown in Theorem \ref{thm:disjointequicliques}. Let the houses in $K \cup K'$ have values $\{a_1, a_2, \dots, a_{|K \cup K'|}\}$ such that $a_1 < a_2 < \dots < a_{|K \cup K'|}$. Since each house has a unique value, we refer to houses using their values for the rest of this proof.

By our assumptions, the houses allocated to $K$ must be split. Therefore there must be some houses that are better than the houses allocated to some nodes in $K$ and worse than houses allocated to other nodes in $K$. This can be formalized as follows
\begin{align*}
    &\pi(a_j) \in K' \text{ for all } j \in [\ell] \text{ and some } \ell \ge 0 \\
    &\pi(a_{l+j}) \in K \text{ for all } j \in [m] \text{ and some } m > 0 \\
    &\pi(a_{l+ m + j}) \in K' \text{ for all } j \in [k] \text{ and some } k > 0 \\
    &\pi(a_{l+m+k+1}) \in K
\end{align*}

We also define the following useful notation: for any clique $K$, we refer to $n_{K, \pi}^{<}(x)$ as the number of agents in $K$ who receive a value strictly less than $x$ under $\pi$. We similarly define $n_{K, \pi}^{>}(x)$.

Construct the allocation $\pi'$ starting at $\pi$ and swapping the houses $a_{l+m+k}$ and $a_{l+m+k+1}$ are swapped. For any node in $C_1$ whose value is less than $a_{l+m+k+1}$ under $\pi$, their envy decreases by $a_{l+m+k+1} - a_{l+m+k}$ in $\pi'$. For any node in $K$ whose value is greater than $a_{l+m+k+1}$ under $\pi$, their envy decreases by $a_{l+m+k+1} - a_{l+m+k}$ in $\pi'$. We can show something similar for $K'$. This gives us the total change in envy as
\begin{align*}
    &\Envy(\pi', G) - \Envy(\pi, G) \\
    &= \Envy(\pi', K \cup K') - \Envy(\pi, K \cup K') \\
    &= [n^{<}_{K', \pi}(a_{l+m+k}) - n^{>}_{K', \pi}(a_{l+m+k})](a_{l+m+k+1} - a_{l+m+k}) \\
    &\qquad \qquad + [n^{>}_{K, \pi}(a_{l+m+k+1}) - n^{<}_{K, \pi}(a_{l+m+l+1})](a_{l+m+k+1} - a_{l+m+k}) \\
    &= [n^{<}_{K', \pi}(a_{l+m+k}) - n^{>}_{K', \pi}(a_{l+m+k}) + n^{>}_{K, \pi}(a_{l+m+k+1}) - n^{<}_{K, \pi}(a_{l+m+l+1})](a_{l+m+k+1} - a_{l+m+k}) \\
    &= [(l+k-1) - (|K'| - l - k)  + (|K| - (m+1)) - m](a_{l+m+k+1} - a_{l+m+k}) \\
    &= [|K| - |K'| + 2(l+k) - 2m - 2](a_{l+m+k+1} - a_{l+m+k})
\end{align*}
Note that due to the optimality of $\pi$, we must have $\Envy(\pi', G) - \Envy(\pi, G) \ge 0$. Since $a_{l+m+k+1} - a_{l+m+k} > 0$ by construction, this implies $|K| - |K'| + 2(l+k) - 2m - 2 \ge 0$. Removing the $-2$, we get $|K| - |K'| + 2(l+k) - 2m  > 0$. This gives us the following observation.

\begin{observation}\label{obs:disconnected-cliques}
$|K'| - |K| - 2(l+k) + 2m  < 0$
\end{observation}

Construct another allocation $\pi''$ as follows: start at $\pi$ and for every $j \in [\min\{m, k\}]$, swap $a_{l+m + 1 - j}$ with $a_{l+m+j}$. In each swap, we swap one house in $K$ with one house in $K'$. Using a similar argument, we can compare the total envy of $\pi''$ and $\pi$.

\begin{align*}
    &\Envy(\pi'', G) - \Envy(\pi, G) \\
    &= \Envy(\pi'', K \cup K') - \Envy(\pi, K \cup K') \notag\\
    &= [n^{<}_{K, \pi}(a_{l+m+1 - \min\{m, k\}}) - n^{>}_{K, \pi}(a_{l+m}) + n^{>}_{K', \pi}(a_{l+m+\min\{m, k\}}) \\
    &\qquad \qquad - n^{<}_{K', \pi}(a_{l+m+1})][\sum_{j \in [\min\{m, k\}]}(a_{l+m+j} - a_{l+m+1-j})] \notag\\
    &= [(m-\min\{m, k\}) - (|K| - m)  \\
    &\qquad \qquad + (|K'| - (l+\min\{k,m\})) - l][\sum_{j \in [\min\{m, k\}]}(a_{l+m+j} - a_{l+m+1-j})] \notag\\
    &= [|K'| - |K| + 2m - 2(\min\{m, k\} + l)][\sum_{j \in [\min\{m, k\}]}(a_{l+m+j} - a_{l+m+1-j})]
\end{align*}\label{eq:disconnected-cliques}
Note that the second term is always strictly positive since $a_{l+m+j} > a_{l+m+1-j}$ for all $j \in \min\{m,k\}$. If we show that the first term $|K'| - |K| + 2m - 2(\min\{m, k\} + l)$ is negative, we contradict the optimality of $\pi$. We have two possible cases.

{\textbf{Case 1: $k \le m$.}} 
In this case, \eqref{eq:disconnected-cliques} reduces to 
\begin{align*}
    \Envy(\pi'', G) - \Envy(\pi, G) = [|K'| - |K| + 2m - 2(k + l)][\sum_{j \in [\min\{m, k\}]}(a_{l+m+j} - a_{l+m+1-j})]
\end{align*}
From Observation \ref{obs:disconnected-cliques}, the first term is negative.

{\textbf{Case 2: $k > m$.}} 
In this case, \eqref{eq:disconnected-cliques} reduces to 
\begin{align*}
    \Envy(\pi'', G) - \Envy(\pi, G) &= [|K'| - |K| + 2m - 2(m + l)][\sum_{j \in [\min\{m, k\}]}(a_{l+m+j} - a_{l+m+1-j})] \\
    &=  [|K'| - |K| - 2l)][\sum_{j \in [\min\{m, k\}]}(a_{l+m+j} - a_{l+m+1-j})]
\end{align*}
Since $|K| > |K'|$ and $l \ge 0$, the first term is negative.

To conclude, it cannot be the case that the houses in $K$ are split.
\end{proof}

\cordisjointcliquesalgorithm*
\begin{proof}
We sort the cliques in a non-increasing order of their size to get $r$ cliques $K^1, K^2, \dots, K^r$ such that $|K^1| \ge |K^2| \ge \dots |K^r|$. 
From Theorem \ref{thm:disjointcliques}, we know that $K^1$ receives a contiguous set of values in the optimal allocation, subject to which, $K^2$ must receive a contiguous set of values among the remaining houses, and so on. 

This gives us a recursive procedure where we try out all possible contiguous sets of values of size $|K^1|$ to give to $K^1$ and subject to that, we try out all possible contiguous sets of values to give to $K^2$ and so on. We then choose the minimum envy allocation among all the allocations computed this way. From Theorem \ref{thm:disjointcliques}, we can conclude that this allocation is optimal.

The pseudocode is presented in Algorithm \ref{algo:disconnected-cliques}. The algorithm maintains a {\em partial} allocation $\pi$ and updates it using recursive calls. 
\begin{algorithm}
    \caption{Minimum Envy House Allocation on Cliques}
    \label{algo:disconnected-cliques}
    \begin{algorithmic}
        \Procedure{FindMinEnvy}{$\{K^i\}_{i \in [r]}, \pi, v, H$}
            \If{$r = 1$}
                \State Update $\pi$ by allocating the houses in $H$ arbitrarily to agents in $K^1$
                \State $\envy = FindEnvy(\pi, v, G)$
                \State \Return $\envy$, $\pi$
            \Else 
                \State $\pi^{\ast} \gets \emptyset, \envy^{\ast} \gets \infty$
                \For{every $|K^1|$ sized contiguous set of values $S$}
                    \State Update $\pi$ by allocating the houses in $S$ arbitrarily to agents in $K^1$
                    \State $\envy^S, \pi^S =$ \Call{FindMinEnvy}{$\{K^{i+1}\}_{i \in [r-1]}, \pi, v, H \setminus S$}
                    \If{$\envy^S < \envy^{\ast}$}
                        \State $\pi^{\ast} \gets \pi^S$
                        \State $\envy^{\ast} \gets \envy^S$
                    \EndIf
                \EndFor
            \State \Return $\envy^{\ast}$, $\pi^{\ast}$
            \EndIf
            
        \EndProcedure
    \end{algorithmic}
\end{algorithm}

To analyze the time complexity, note that we compute at most $O(n^r)$ allocations. For each allocation, finding the envy of the allocation takes $O(n^2)$ time trivially. This gives us a total time complexity of $O(n^{r+2})$.
\end{proof}

\end{document}